%% file: soda-main.tex
\newif\iffullpaper
\fullpaperfalse
\fullpapertrue

\newcommand{\infull}[2]{
    \iffullpaper
    #1
    \else
    #2
    \fi
}

\infull{
    \documentclass[twoside,leqno]{article}
    \usepackage{amsthm}
    \usepackage{fullpage}
    \newtheorem{theorem}{Theorem}
    
    \newtheorem{lemma}{Lemma}

}{
    \documentclass[twoside,leqno,twocolumn]{article}
    \usepackage{ltexpprt}
}

\usepackage{thm-restate}

\usepackage{microtype}
\usepackage{color}
\usepackage[usenames,dvipsnames,svgnames,table]{xcolor}

\usepackage{mathrsfs}
\usepackage{amssymb,latexsym,amsmath}
\usepackage{mathrsfs}
\usepackage{hyperref}
\usepackage{caption}
\usepackage{subcaption}

\usepackage{epsfig}

\newtheorem{definition}{Definition}

\newtheorem{ob}{Observation}

\include{pre}

\infull{
  \title{\Large A Lower Bound for Dynamic Fractional Cascading}
}{
  \title{\Large A Lower Bound for Dynamic Fractional Cascading\footnote{For the full paper see~\cite{Afshani.dyfc.arxiv}.}}
}
\author{Peyman Afshani\thanks{Supported by DFF (Det Frie Forskningsr\" ad) of Danish Council for Independent Research under grant ID DFF$-$7014$-$00404.}\\ Computer Science Department, Aarhus University\\\texttt{peyman@cs.au.dk}}

\infull{
}{
\fancyfoot[R]{\scriptsize{Copyright \textcopyright\ 2021 by SIAM\\
Unauthorized reproduction of this article is prohibited}}
}

\begin{document}
%\author{Peyman Afshani \\ MADALGO\thanks{MADALGO is a center of the Danish National Research Foundation.}, Aarhus University \\ \texttt{peyman@madalgo.au.dk} }

\date{}
\maketitle
\begin{abstract}
We investigate the limits of one of the fundamental ideas in data structures: fractional
	cascading. This is an important data structure technique to speed up repeated searches 
  for the same key in multiple lists and it has numerous applications. 
	Specifically, the input is a ``catalog'' graph, $\cat$, of constant degree together
  with a  list of values assigned to every vertex of $\cat$.
  The goal is to preprocess the input such that given a
  connected subgraph $\scat$ of $\cat$ and a single query value $q$,
  one can find the predecessor of $q$ in every list that belongs to $\scat$.
	The classical result by Chazelle and Guibas shows that in a pointer machine, this can be done in the optimal
	time of $\O(\log n + |\scat|)$ where $n$ is the total number of values. 
	However, if insertion and deletion of values are allowed, then the query time slows down to 
	$\O(\log n + |\scat| \log\log n)$.
	If only insertions (or deletions) are allowed, then once again,
	an optimal query time can be obtained but by using amortization at update time.

  We prove a lower bound of $\Omega( \log n \sqrt{\log\log n})$ on the
  worst-case query time of dynamic fractional
  cascading, when queries are paths of length $O(\log n)$.
  The lower bound applies both to fully dynamic data structures with
  amortized polylogarithmic update time and 
  incremental data structures with polylogarithmic worst-case update time.
  As a side, this also proves that amortization is crucial for obtaining an optimal incremental data structure.
  %As another side, we can show that the dynamic rectangle stabbing
  %problem in the pointer machine model requires $\Omega(\log n \sqrt{\log\log n})$ query time,
  %assuming polylogarithmic update time. Note that the static problem can be solved
  %in linear space and with $O(\log n + k)$ query time. 

    This is the first non-trivial pointer machine lower bound
    for a dynamic data structure that breaks the $\Omega(\log n)$ barrier.
    In order to obtain this result, we develop a number of new ideas and techniques that hopefully can be
    useful to obtain additional dynamic lower bounds in the pointer machine model.
\end{abstract}

\input{intro}

\input{prel}

%include{CUSF}
%\input{summary}
\input{reduction-short}

\input{new}

\section{Acknowledgements}
The author would like to thank Seth Pettie for the encouragements
he provided throughout the years he has been working on this problem.

\bibliographystyle{plain}
\bibliography{ref}

\infull{
\appendix
\input{reduction}

}{
}

\end{document}

%% file: pre.tex
\usepackage{yfonts}

\usepackage{soul}

\newcommand{\ignore}[1]{} 

\newcommand{\refe}[1]{Equation~(\ref{eq:#1})}

\newcommand{\bL}{\mathbf{L}}
\renewcommand{\S}{\mathcal{S}}

\newcommand{\bphi}{\mathbf{\Phi}}

\newcommand{\M}{\mathcal{M}}
\newcommand{\U}{\mathcal{U}}
\newcommand{\Q}{\mathcal{Q}}
\newcommand{\F}{\mathcal{F}}

\newcommand{\C}{\mmfont{C}}

\newcommand{\cT}{\cat}

\newcommand{\E}{\mathbb{E}}
\newcommand{\N}{\mathbb{N}}

\newcommand{\cS}{\mathcal{N}}
\newcommand{\R}{\mathbb{R}}

\newcommand{\A}{\mathscr{A}}

\newcommand{\LL}{\log \log}

\newcommand{\deff}[2]{
  \begin{definition}\label{def:#1}
    #2
  \end{definition}
}
\newcommand{\refd}[1]{Definition~\ref{def:#1}}

\newcommand{\obs}[2]{
  \begin{ob}\label{ob:#1}
    #2
  \end{ob}
}

\newcommand{\lemm}[2]{
  \begin{lemma}\label{lem:#1}
    #2
  \end{lemma}
}
\newcommand{\refl}[1]{Lemma~\ref{lem:#1}}
\newcommand{\refo}[1]{Observation~\ref{ob:#1}}

\newcommand{\refq}[1]{Equation (\ref{eq:#1})}

\newcommand{\out}{\mbox{Output}}
\newcommand{\young}[1]{\mbox{Y}\left(#1\right)}
\newcommand{\spurt}[1]{\mbox{Spurt}\left(#1\right)}

\newcommand{\enc}[1]{\mbox{Enc}\left(#1\right)}

\newcommand{\idef}[1]{{{\underline{#1}}}}
\newcommand{\mdef}[1]{{{\underline{#1}}}}

\newcommand{\cat}{\mathbf G_{\mbox{\tiny cat}}}
\newcommand{\scat}{\mathbf G'}
\DeclareMathAlphabet{\mathpzc}{OT1}{pzc}{m}{it}
\newcommand{\mmfont}[1]{\mathpzc{#1}}
\newcommand{\mem}{\mmfont{G}_{\mbox{\tiny mem}}}

\newcommand{\B}{\mmfont{B}}

\newcommand{\cu}{\mmfont{u}}
\newcommand{\cv}{\mmfont{v}}

\newcommand{\cw}{\mmfont{w}}

\renewcommand{\O}{\mathcal{O}}

\newcommand{\area}{\mbox{A{\scriptsize rea}}}
\newcommand{\region}{\mbox{R{\scriptsize egion}}}
\newcommand{\dead}{\mbox{D{\scriptsize ead}}}

\newcommand{\reg}{\region}
\newcommand{\dense}{\mbox{D{\scriptsize ense}}}

\newcommand{\regset}{\mbox{RegionSet}}
\newcommand{\areasum}{\mbox{A{\scriptsize reaSum}}}
\newcommand{\asum}{\areasum}

\def\pT{\mathscr{T}}
\def\aT{\mathbb{T}}

\def\B{\mathcal{B}}

\def\R{\mathcal{R}}

\newcommand{\eoc}{\mbox{EoC}}

%% file: intro.tex
\section{Introduction}
Our motivation lies at the intersection of two important topics: the fractional
cascading problem and proving dynamic lower bounds in the pointer machine model. 
We delve into each of them below but to summarize, we give the first lower bound for the
fractional cascading problem, which we believe is the first major progress on this problem
since 1988, and in addition, our lower bound is the first non-trivial pointer machine lower bound
for a dynamic data structure that breaks the $\Omega(\log n)$ barrier.

By now fractional cascading is one of the fundamental and classical techniques in data structures and in 
fact, it is routinely taught in various advanced data structures courses around the world.
Its importance is due to a very satisfying and elegant answer that it gives
to a common problem in data structures: 
how to search for the same key in many different lists? 
Fractional cascading provides a very general framework to solve the problem: 
the input can be any graph of constant degree, $\cat$,  called the ``catalog'' graph.
Also as part of the input, each vertex of $\cat$ is associated with a ``catalog'' that is simply
a list of values from an ordered set.
The goal is to build a data structure such
that given a query value $q$ from the ordered set, 
and a connected subgraph $\scat$ of $\cat$, one can find the predecessor of $q$
in every catalog associated with the vertices of $\scat$.
As the lists of different vertices are unrelated, at the first glance it seems
difficult to do anything other than just a binary search in each list,
for a total query time of $\Omega(|\scat|\log n)$.
Fractional cascading reduces this time to $O(|\scat| + \log n)$, 
equivalent to constant time
per predecessor search, after investing an initial $\log n$ time. 

This problem of performing iterative searches has shown up multiple times in the past. 
For example, in 1982 Vaishnai and Wood~\cite{Vaishnavi} used ``layered segment
tree'' to break-through this barrier and 
in 1985 research on the planar version of orthogonal range reporting led Willard~\cite{Willard85} 
to the notion of ``downpointers'' that allowed him to do $\log n$ searches in $O(\log n)$ time.
The next year and in a two-part work, Chazelle and Guibas presented the fully-fledged framework of fractional
cascading and used it to attack a number of very important problems in data structures~
\cite{Chazelle.Guibas.fractional.I,Chazelle.Guibas.fractional.II}:
they presented a linear-size data structure that could answer fractional cascading queries in 
$O(\log n + |\scat|)$ time. 
As discussed, this is optimal in the comparison model but also in the pointer machine model.
The idea behind fractional cascading also shows up in other areas, e.g., 
some of the important milestones in parallel algorithms were made possible by similar ideas
(e.g., Cole's seminal $O(\log n)$ time
parallel sorting algorithm~\cite{colesort}; see also~\cite{cascading89} for further
applications in the PRAM model).

\subparagraph{The dynamic case.}
However, despite its importance, the dynamic version of the problem is still open. 
Chazelle and Guibas themselves investigated this variant.
Here, some optimal results were obtained quickly:
If only insertions (or deletions) are allowed, then updates can be done in $O(\log n)$ amortized time (i.e.,
in $O(n\log n)$ time over a sequence of $n$ updates) while keeping the optimal query time.
However, if both insertions and deletions are allowed then 
the query time becomes $O(\log n + |\scat|\log\log n)$. 
After presenting the dynamic case, Chazelle and Guibas expressed dissatisfaction at their
solution, wondering whether it is 
optimal~\footnote{They write~\cite{Chazelle.Guibas.fractional.I}: ``The most unsatisfactory
  aspect of our treatment of fractional cascading is the handling of the
dynamic situation. Is our method optimal?.''}.
Later work by Mehlhorn and N\" aher~\cite{mn90} improved this dynamic
solution by cutting down update time to amortized $O(\log\log n)$ time but they could not offer any improvement on the query
time.
Dietz and Raman removed the amortization to make the update time of $O(\log\log n)$ worst-case~\cite{DietzRaman91}. 
Finally, while working on the dynamic vertical ray shooting problem, 
Giora and Kaplan~\cite{GioraKaplan09} 
showed that the extra $\log\log n$ factor is in fact an additive term when considering the degree
of the graph $\cat$, i.e., queries can be performed in $O(\log n + |\scat|(\log\log n + \log d))$ time where
$d$ is the maximum degree of $\cat$.
However, despite all this attention given to the dynamic version of the problem, the extra $\LL n$ factor behind the output size persisted.

The results by Chazelle and Guibas hold in the general pointer machine model: the memory of the data structure
is composed of cells where each cell can store one element as well as two pointers to other memory cells and  the
only way to access a memory cell is to follow a pointer that points to the memory cell. 
In this model, it seems difficult to improve or remove the extra $\log\log n$ factor. 
On the other hand, lower bound techniques for dynamic data structure problems in the pointer machine model
are quite under-developed.
We will discuss these next.

\subparagraph{The pointer machine lower bounds.}
As it will be evidence soon, proving \textit{dynamic} lower bounds in this model is very challenging, 
even though it is 
one of the classical models of computation and it is a 
very popular model to represent tree-based
data structures as well as many data structures in computational geometry.
Here, our focus will be on the lower bound results only.

Many fundamental algorithmic and data structure problems were studied in 1970s and 
in the pointer machine model. 
Perhaps the most celebrated lower bound result of
this period is Tarjan's $\Omega(m\alpha(m,n))$ lower bound for the complexity
of $n$ ``union'' operations and $m$ ``find'' operations when
$m>n$,~\cite{Tarjan-UF-79}. 
This was later improved to
$\Omega(n + m\alpha(m,n))$~\cite{Ban80,LaP96} for any $m$ and $n$.
Here $\alpha(m,n)$ is the inverse Ackermann function.
However, originally, the existing lower bounds needed a certain ``separation
assumption'' that essentially disallowed placing pointers between some elements
(e.g., elements from different
sets)~\cite{Tarjan-UF-79,TarjanvanLeeuwen84,Ban80}.  
In 1988, Mehlhorn et al.~\cite{Union.Split.LB.88} studied the dynamic predecessor search problem in the pointer
machine model, under the name of ``union-split-find'' problem,  and they
proved that for large enough $m$, a sequence of $m$ insertions, $m$ queries
and $m$ deletions, requires $\Omega(m \log\log n)$ time.  This was a
significant contribution since not only they proved this without using the
separation assumption but also in the same paper they showed that the
separation assumption can in fact result in loss of efficiency, i.e., a pointer
machine without the separation assumption can outperform a pointer machine with
the separation assumption.
Following this and in 1996, La Poutre  showed that lower bounds for the union-find problem still hold
without the separation assumption~\cite{LaP96}.
We note that the lower bound of  Mehlhorn et al. was later strengthened by Mulzer~\cite{Mulzer-USF-09}.

The lower bound of Mehlhorn et al.~\cite{Union.Split.LB.88} on the dynamic predecessor search problem can be
viewed as a  ``budge fractional cascading lower bound''.
However, it only provides a very limited lower bound;
essentially, it only applies to data structures that treat a dynamic fractional cascading query on
a subgraph $\scat$ as a set of $|\scat|$ independent dynamic predecessor queries. 
But obviously, the data structure may opt to do something different and in fact in some cases
something different is actually possible. 
For example, if $\cat$ is a path, then the dynamic fractional problem on $\cat$ reduces to
the dynamic interval stabbing problem for which data structures with $O(\log n + |\scat|)$ query time exist
(e.g., using the classical segment tree data structure).
As a result, the lower bound of Mehlhorn et al.~\cite{Union.Split.LB.88} is not enough to rule out
``clever'' solutions that somehow circumvent treating the problem as a union of independent predecessor queries.

As far as we know, these are the major works on dynamic lower bounds in the pointer machine model.
It turns out, proving non-trivial lower bounds in this model is in fact quite challenging.
In all the dynamic lower bounds above, after being given a query, the data structure has $o(\log n)$
time to answer it: 
in the dynamic predecessor problem studied by Mehlhorn et al.~\cite{Union.Split.LB.88},
the query is a pointer to an element $u$ and the data structure is to find the predecessor
of $u$ in $O(\log\log n)$ time. 
In the union-find problems, the query is once again a pointer to an element of a set, and the
data structure should find the ``label'' of the element with very few pointer navigations.
Contrast this with the fractional cascading problem: the query could be a subgraph $\scat$ of
size $\log n$, which would give the data structure at least $\Omega(\log n)$ time.

Lack of techniques for proving high pointer machine lower bounds for dynamic
data structure is quite disappointing, specially compared to 
the static world where there are a lot of
lower bounds to cite and at least two relatively easy-to-use lower bound frameworks. 
In fact, static pointer machine lower bounds are so versatile that they have been adopted
to work in other models of computation, such as the I/O model.
But not much has happened for a long time in the dynamic front.

\subsection{Our Results.}
We believe 
we have made significant progress in two fronts: 
we prove a lower bound of $\Omega(
\log n \sqrt{\log\log n})$ on the worst-case query time of dynamic fractional
cascading, when queries are paths of length $O(\log n)$.
Our lower bound actually is applicable in two scenarios:
(i)  when the data structure is fully dynamic and the update time is amortized and polylogarithmic
(i.e., it takes $O(n \log^{O(1)} n)$ time to perform any sequence of $n$ insertions or deletions)
or (ii) in the incremental case, when the data structure is only required to do insertions
but the update time must be polylogarithmic and worst-case.
This proves that in an incremental data structure, amortization is required to keep the query time optimal, inline
with the upper bounds~\cite{Chazelle.Guibas.fractional.I}.

As far as we are aware, this is the first non-trivial\footnote{
    By ``non-trivial'' we mean a query lower bound that is asymptotically higher than 
    the best lower bounds for the static version of the problem;
    clearly, any lower bound that one can prove for a static data structure problem,
    trivially applies to the dynamic case as well.} pointer machine lower bound for a dynamic data structure 
that exceeds $\Omega(\log n)$ bound.
Thus, we believe our technical contributions are also important.
We introduce a number of ideas that up to our knowledge are new in this area.
Unfortunately, our proof is quite technical and involves making a lot of definitions and 
small observations. 
Simplifying our techniques and turning them into a more easily applicable framework
is an interesting open problem.

Finally, we remark that since we obtain our lower bound using a geometric construction,
our results have another consequence: we can show that the dynamic rectangle stabbing
problem in the pointer machine model requires $\Omega(\log n\sqrt{ \log\log n})$ query time,
assuming polylogarithmic update time. The static problem can be solved
in linear space and with $O(\log n + k)$ query time~\cite{c86} and thus our lower bound
is the first provable separation result between dynamic vs static versions of a 
range range reporting problem.

%% file: prel.tex
\section{The Model of Computation and Known Static Results} 
\subsection{The Lower Bound Model.}
We now formally introduce the particular variant of the pointer machine model that is
used for proving lower bounds. 
Here, the memory is composed of \idef{cells} and each cell can store one value from the
catalogue of any vertex of $\cat$.  
In addition, each cell can store up to two pointers to other memory cells.  
We think of the memory of the data structure as a directed graph with
out-degree at most two where a pointer from a memory cell $\cu$ to a memory cell $\cw$ is
represented as a directed edge from $\cu$ to $\cw$.  
There is a special cell, the root, that the data structure can
always access for free.

We place two restrictions in front of the data structure: one, the only way to access
a memory cell $\cw$ is first to access a memory cell $\cu$ that points to $\cw$ and then 
to use the pointer from $\cu$ to $\cw$ to access $\cw$.
And two, the input values must be stored atomically by the data structure and the only way
the data structure can output any input value $e_i$ is to access a memory cell $\cu$ that stores
$e_i$.

On the other hand, we only charge for two operations: At the query time, we only charge for pointer
navigations, i.e., only count how many cells the data structure accesses. 
At the update time, we only count the number of cells that the data structure changes.

Thus, in effect we grant the data structure infinite computational power, as well
as full information regarding the structure of the memory graph; e.g., the query algorithm
fully knows where each input value is stored and it 
can compute the best possible way (i.e., solve an implicit Steiner subgraph problem) to reach
a cell that stores a desired output value or the update algorithm can figure out how to change
the minimum number of pointers to allow the query algorithm the best possible ability 
to do the navigation.
In essence, a dynamic lower bound is a statement about ``the connectivity bottleneck''
of a dynamic directed graph with out-degree two.

\subsection{Static Lower Bounds.}
Lower bounds for static data structure problems in the pointer machine model have a very good
``success rate'', meaning, there are many problems for which these lower bounds match or almost
match the best known upper bounds. 
We can attribute  the initiation of this line of research to Bernard Chazelle,
who in 1990~\cite{Chazelle.LB.reporting}
introduced the first framework for proving lower bounds for static data structures in the pointer 
machine model and also used it to prove an optimal space lower bound for the important
orthogonal range reporting problem.
Since then, for other important problems, similar lower bounds have been found: 
optimal and almost optimal space-time trade-offs for the fundamental
simplex range reporting problem~\cite{a12,Chazelle.simplex.RR.LB}, optimal query lower bounds
for variants of ``two-dimensional fractional cascading''~\cite{cl04,AC.2dfc}, 
optimal query lower bounds for the axis-aligned rectangle stabbing problem~\cite{aal10,aal12}, and
lower bounds for multi-level range trees~\cite{AD.frechet17}.
This list does not include lower bounds that can be obtained as consequences of these lower bounds
(e.g., through reductions).
In addition, the pointer machine lower bound frameworks have been applied to some string 
problems~\cite{Asfhanietal.jumbled20,afshaniNielsen16,Cohen-Addadetal18} as well as  to
the I/O model (a.k.a the external memory model)~\cite{indexmodel,refinedred}.
In fact, the author has shown that under some very general settings,
it is possible to directly translate a pointer machine lower bound to a lower bound 
for the same problem in the I/O model~\cite{a12}.

Unfortunately, none of the above progress can be translated to dynamic problems, mostly
due to a severe lack of techniques for proving dynamic pointer machine lower bounds.
It is our hope that this paper in combination with the techniques used in the aforementioned static lower bounds 
can be used to advance our knowledge of dynamic data structure lower 
bounds. 

\ignore{
~\footnote{
  As mentioned earlier, in this work, we already make some progress, by showing that 
  for the rectangle stabbing problem,
  the dynamic version is provably more difficult to solve than the
  static problem.}.
}

\section{Preliminaries}
We will always be dealing with directed graphs, when are talking about structures in the memory 
and thus we may drop the word ``directed'' in this context. Here, a directed a tree is one
where all the edges are directed away from root. 
Also, as already mentioned, we grant unlimited computational
power and full information about the current status of the memory to the algorithm.
But this does not apply to the future updates!
\textit{The algorithm does not know what will be the
future updates, in fact, 
 revealing that information will cause the lower bound to disappear.}

For a vertex
$\cv$ in a directed graph, \idef{$k$-in-neighborhood of $\cv$} is the set
of vertices that have a directed path of length at most $k$ to $\cv$ and
\idef{$k$-out-neighborhood of $\cv$} is the set of vertices that
can be reached from $\cv$ by a directed path of size at most $k$.

\subsection{The structure of the proof.}
In the next section, we show a general reduction that allows us to apply a 
lower bound for incremental data structures with worst-case update time,
to fully dynamic data structures with amortized update time. 
Thus, in the rest of the proof we only consider incremental data structures with 
polylogarithmic worst-case update times. 

In Section~\ref{sec:input}, we construct a set of difficult insertions.
Our catalog graph is a balanced binary tree of height $\log n$. 
Then, to define the sequence of insertions, 
we work in $t$ epochs and in the $i$-th epoch 
we will insert $O(\frac{bn}{2^{\gamma i}})$ values into the catalog graph, 
for some parameter $b, t$ and $\gamma$; their exact values don't matter but note that
this is a geometrically decreasing sequence of $i$, the epoch number. 
The values that we choose to insert are selected randomly with respect to a particular distribution; 
it is important that the data structure does not know which values are going to be inserted~\footnote{
  Otherwise, it can ``prepare'' for the future updates, meaning, it can build the necessary ``connectivity
  structures'' to accommodate those updates and then just update very few cells once those values 
get inserted.}.
In addition, in this section we show that the problem can be turned into a geometric stabbing
problem where  inserted values can be turned into  rectangles (called sub-rectangles) and a fractional
cascading query is mapped to a point. The answer to the query is the set of sub-rectangles that contain
the query point. 

In Section~\ref{sec:pers}, we define a notion of ``persistent structures''.
We call them ordinary connectors and their properties are roughly summarized
below: an ordinary connector $\C$ (i) is a subtree of the memory graph at the end of epoch $t$
(ii) some of its cells are marked as ``output nodes'' (i.e., output producing nodes)
and each output node stores  a sub-rectangle (i.e., a value inserted into some catalog) 
(iii) for every output node $\cw$ that stores a sub-rectangle (i.e., value) inserted during epoch $i$,
no ancestor of $\cw$ is updated in epochs $i+1$ or later. 
Recall that we only ask queries at the end of epoch $t$ and so 
the main result of this section is that if the query time is smaller than the claimed lower bound, then
the data structure much explore a subgraph of the memory graph which contains ``large'' ordinary connectors with a 
``high density'' of output nodes;
thus, if the worst-case query time is smaller than the claimed lower bound,  then all the queries
must have this property. 
An interesting aspect of this proof is that it does not use the fact that the data structure
does not know about the future updates, meaning, in this section we can even assume that the data
structure knows exactly what values are going to be inserted in each epoch. 

Section~\ref{sec:lb} is the most technical part of the proof and it tries to reach a contradiction
in light of the main theorem of Section~\ref{sec:pers}. 
Recall that Section~\ref{sec:pers} had proven that $\C$ must be ``large'' and ``dense'', i.e.,
it  must have a lot of cells and a large fraction of its cells should be output cells.
This means that as we look at the evolution of $\C$ throughout the epochs, it should collect
a lot of output cells at each epoch. 
This brings us to a notion of living tree. These are subtrees of the memory graph at some epoch $i$
with the property that they can potentially ``grow'' to be an ordinary connector $\C$; another
way of saying it is that living trees of epoch $i$ are subgraphs of $\C$ during epoch $i$. 
Next,  we allocate a notion of ``area'' or ``region'' to each ordinary connector as well as to each living tree.
Recall that in Section~\ref{sec:input} we have shown that  fractional cascading queries can be
represented as points. 
Thus, intuitively, a region associated to an ordinary connector $\C$ (or a living tree) is a subset of the query region that
contains the points for which $\C$ can be the ordinary connector (or grow to be the regular connector) 
claimed by the main theorem of Section~\ref{sec:pers}.
Using this concept, then we associate a notion of ``area'' and ``region'' to each memory cell $\cv$.
Roughly speaking, this represents the regions of all the living trees for which $\cv$ can be updated
to be a living tree in the future epochs. 
In this section, we have to use the critical limitation of the algorithm which is that it does not
know the future updates. 
Using this and via a potential function argument, we essentially show that the probability that a living tree
in epoch $i$ grows into a living tree in future epochs is small. 
As a result, we can bound the expected number of living trees in each epoch.
As a further consequence, this bounds the expected number of regular connectors at the end of epoch $t$.
By picking our parameters carefully, we make sure that this number is sufficiently small that it contradicts
the claim in Section~\ref{sec:pers}.
The only way out of this contradiction is that some queries times must be larger than our claimed lower bound, thus,
proving a lower bound on the worst-case query time.

%% file: reduction-short.tex
\section{From Worst-case to Amortized Lower Bounds}\label{sec:reduction}

%\para{The definition of amortization.}
We work with the following definition of amortization.
We say that an algorithm or data structure has an amortized cost of $f(n)$, for a function $f:\N \to \N$,
if for any sequence of $n$ operations, the total time of performing the sequence 
is at most $nf(n)$.

%\para{Epoch-based worst-case incremental adversary.}
%Assume, we are considering an incremental data structure.
We call the following adversary, an Epoch-Based Worst-Case Incremental Adversary (EWIA) with update
restriction $U(n)$; here $U(n)$ is an increasing function.
The adversary works as follows.
We begin with an empty data structure containing no elements and then 
the adversary reveals an integer $k$ and they announce that they will insert $O(n)$ elements
over $k$ epochs.
Next, the adversary allows the data
structure $nU(n)$ time before anything is inserted.
At each epoch $i$, they reveal an insertion sequence, $s_i$, of size $n_i$. %, where $n_i < n_{i-1}$.
At the end of epoch $k$, the adversary will ask one query. 
The only restriction 
is that the insertions of epoch $i$ must be done in $n_i U(n)$ time once $s_i$ is revealed.
So the algorithm is not forced to have an actually worst-case insertion time and it suffices
to perform all the $s_i$ insertions in $n_i U(n)$ time. 
We iterate that the adversary allows the data structure to operate in the stronger pointer machine model 
(i.e., with infinite computational power and full information about the current status of the memory graph).

\begin{restatable}{lemma}{reduction}\label{lem:ii2i}
    Consider a dynamic data structure problem where we have insertions, deletions and queries.
    %and assume the following holds for any data structure solving this problem.
    Assume that we can prove a worst-case query time lower bound of $Q(n)$ for any data structure,
    using an EWIA with $k$ epochs and with update restriction $U(n)$.
        %Furthermore, assume the adversary uses $k$ epochs where at epoch $i$
    %the adversary inserts $n_i$ elements.
    %Let $n= n_1 + \dots n_k$.

    Then, any fully dynamic data structure $\A$ that can  perform any sequence of $N$ insertions and deletions,
    $N \ge n$,
    in $\frac{U(N)}{8k}N$ total time (i.e., $\A$ has amortized $O(U(N)/k)$ update time),
    must also have $\Omega(Q(n))$ lower bound for its worst-case query time. 
\end{restatable}
\begin{proof}
\infull{
  See Appendix~\ref{app:reduction}.
  }{
      See the full version~\cite{Afshani.dyfc.arxiv}.
  }
\end{proof}

%% file: new.tex
\section{The Input Construction}\label{sec:input}
In this section, we describe our catalog graph as well as the sequences of insertions
used in our lower bound.

The catalog graph $\cat$ is a balanced binary tree of height $\log n$.
To be able to describe the sequence of insertions, we will use a geometric
representation of our construction as follows. 
We measure the level of any vertex in $\cat$ from 1, for convenience. 
So, the root of $\cat$ has level $1$ while any vertex with distance $k$ to 
the root  has level $k+1$; the leaves of $\cat$ have level $\log n$.
Let $\Q = [1,n] \times [1,n]$.
We use the $Y$-axis to model the universe (i.e., the values in the catalogs)
while the $X$-axis will be used to capture the catalog graph $\cT$.
We map each vertex in $\cT$ to a rectangular region in $\Q$. 
Partition  $\Q$ into $2^{j-1}$ congruent ``slabs'' using vertical lines
(i.e., $2^{j-1}-1$ vertical lines). 
Assign these rectangles to the vertices of $\cat$ at level $j$, from left to
right.
In particular, the root of $\cat$ will be assigned the entire $\Q$
whereas the left child of the root will be assigned the left half of
$\Q$ and so on. 
For a vertex $v \in \cat$, let $R(v)$ be the rectangle assigned to $v$.
See Figure~\ref{fig:geo}(left).

\begin{figure*}[h]
    \centering
    \includegraphics[scale=0.50]{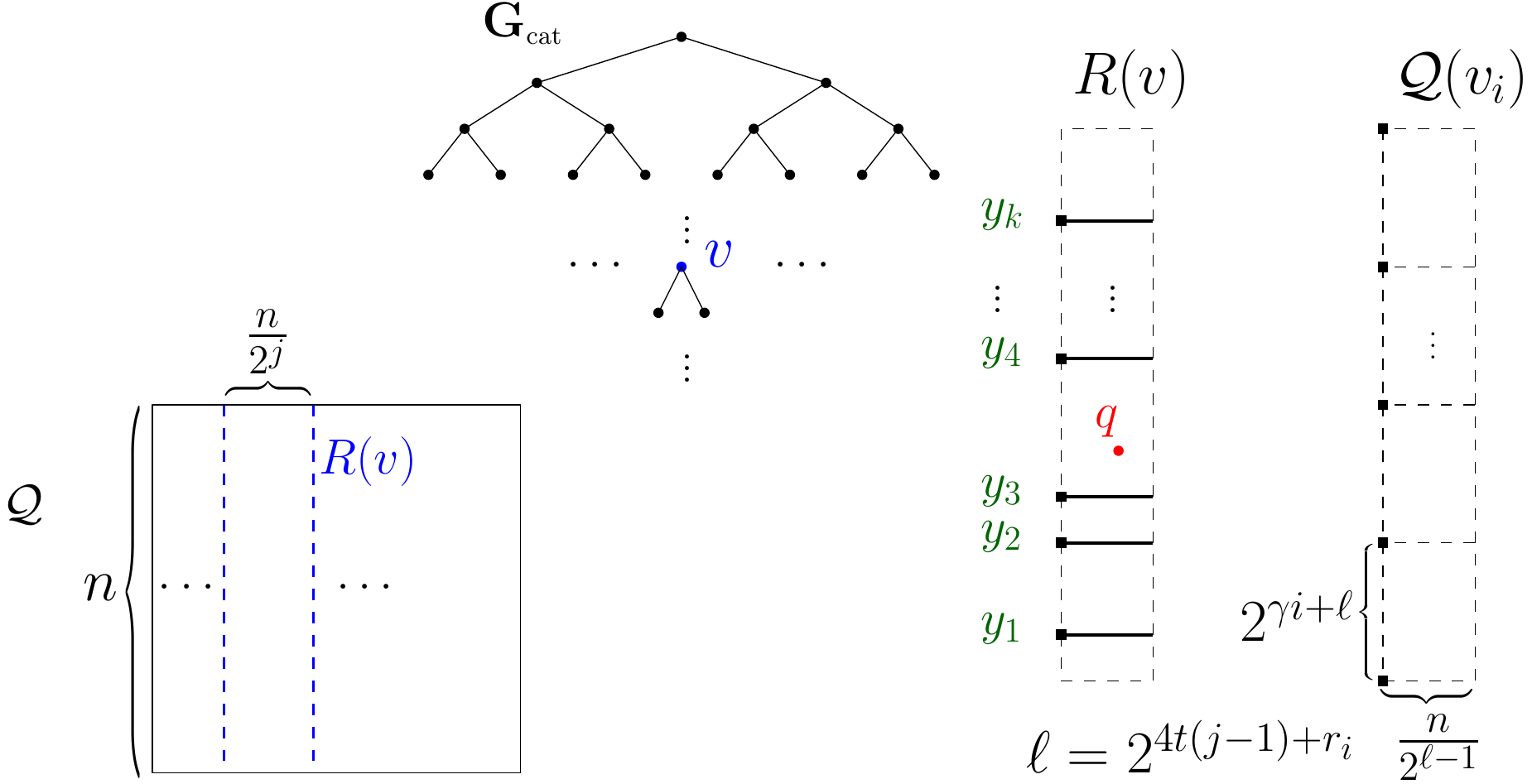}
    \caption{}
    \label{fig:geo}
\end{figure*}

\begin{ob}\label{ob:ancestor}
  For two vertices $v,w \in \cat$, 
  $R(v)$ and $R(w)$ have non-empty intersection
  if and only if one of the vertices is an ancestor of the other one.
\end{ob}

We will only insert integers between $1$ and $n$ into the catalog
of vertices of $\cat$.
Assume, we have inserted values $y_1, \dots, y_k \in [n]$ into the catalog of
$v$.
Then, in our geometric view, we partition $R(v)$ into \idef{sub-rectangles}
using horizontal lines at coordinates $y_1, \dots, y_k$.
See Figure~\ref{fig:geo}(right).

\begin{ob}
  Let $\pi$ be a path from the root of $\cat$ to a leaf $v$. Let $y$ be a (query) value in $[n]$.
  Consider the pair $(\pi,y)$ as a fractional cascading query.

  Let $x$ be the $X$-coordinate of a point that the lies inside $R(v)$. 
  Then, the fractional cascading query is equivalent to finding all the sub-rectangles that
  contain the query point $q=(x,y)$. 
\end{ob}
%\begin{proof}
%  Observe that only sub-rectangles of the nodes on the path $\pi$ can contain
%  the $X$-coordinate of the point $(x,y)$. 
%  Finally, no
%\end{proof}

In the rest of our proof, we will only consider fractional cascading queries that can be
represented by such a query point $q$. 
As we are aiming to prove a lower bound, this restriction is valid.  
Observe that in this view, the fractional cascading problem is essentially a (restricted form of) rectangle stabbing problem. 
We are now ready to define our sequence of insertions. 

\subsection{Parameters and Assumptions.}
We will only consider data structures with polylogarithmic update times. 
Let $U$ be the update time of the data structure.  Thus, $\log U = \Theta(\log\log n)$.
Shortly, we will describe our construction of difficult insertion sequences followed by a proof of why they are difficult.
This will be done using a number of parameters; it is possible to read the proof without knowing the parameters in advance
(in particular, at some point the author had to select the best parameters to obtain the highest lower bound), but probably
it is more convenient to have their final values in mind
so that one can have a ``ballpark'' notion of how large each parameter is.
Let $\gamma = 100 (\log \log n + \log U)$, $t = \frac{\log n}{\gamma}$ and $b=\gamma/4$.
%Define $t=(\log n)^{0.9}$.
%Let $t < \log n$ be a parameter that will be determined later; our assumption for now is that
%$t$ will be such $\log t = \Omega(\log\log n)$.
%Another parameter $b$ will be picked such that $tb= \frac{\log n}{4}$.
$b$ and $t$ will be used in the construction, specifically, to define our insertion sequences. 
Observe that $tb= (\log n)/4$.

The list of other parameters that we use include the following. 
We set $\alpha =\varepsilon^2  \sqrt{\log\log n}$, %$\alpha_2 = \alpha/\delta$
and $\beta = \varepsilon\sqrt{\log\log n}$ and $\varepsilon$ is a sufficiently small constant.
Our goal is to show a lower bound of $\Omega(\alpha \log n)$ for a fractional cascading query of length
$\log n$, meaning, $\alpha$ will the multiplicative factor of our lower bound. 
%In particular, we will show that dynamic fractional cascading on a path of length $\log n$ requires
%$\Omega(\alpha \log\log n)$ time. 
We  define $\mu = \frac{\log t}{8}$ and  
 $c =  \frac{t\gamma}{18\alpha \beta} \ge 4t$.

\subsection{Details of the Construction}
Our construction has $t$ epochs where at epoch $i$, we will insert some values in the
catalog of some vertices of $\cat$. 
The value $t$ is fixed and known in advance, by the data structure.

We partition the vertices of $\cat$  into subsets that we call \mdef{bands}.
The first $2t$ levels of $\cat$ (i.e., the vertices in the top $2t$ levels of
$\cat$) are placed in the first band, the next $2t$ levels are placed in the
second band and so on. 
Overall, we get $\frac{\log n}{2t}$ bands (for simplicity, we assume the division
is an integer).

The vertices of the even bands (i.e., bands 2, 4, 6, etc.) will have their
catalogs empty; we will never insert anything in vertices of these bands. 
Let $b$ be the number of odd bands. 

To be more precise, the $j$-th odd band consists of levels $4t(j-1)+1$ to $4t(j-1)+2t$.
Initially all these levels are marked as \idef{untaken}.
During epoch $i$, we select $b$ untaken levels, $\ell_{i,1}, \dots, \ell_{i,b}$ where $\ell_{i,j}$ is selected
uniformly at random among all the untaken levels of the $j$-th odd band. 
These levels then are marked as \idef{taken}; this operation also comes with
corresponding insertion sequences that we will define shortly.
However, we have the following observation.
\begin{ob}\label{ob:choices}
  In every epoch, and in every odd band,
  we have at most $2t$ choices and at least $t$ choices for which level to
  take.
  Furthermore, the choices of different odd bands are independent. 
\end{ob}

Assume we are at epoch $i$.
Then, taking a level $\ell = 4t(j-1)+r_i$ uniquely determines the sequence of
values that we insert in the cataglogs of the vertices of level $\ell$, as follows.
Consider a vertex $v$ at level $\ell$.
We divide $R(v)$ into $\frac{n}{2^{\gamma i + \ell}}$ congruent sub-rectangles which means
we insert the $Y$-coordinates of their boundaries as values in the catalog of 
$v$ (Figure~\ref{fig:geo}(right)).

\begin{lemma}\label{lem:inssize}
  During epoch $i$, a total of $\frac{bn}{2^{\gamma i}}$ values are inserted,
  independent of which levels are taken.
  Also, a sub-rectanble inserted during epoch $i$ at a vertex of level $\ell$ has 
  height $2^{\gamma i + \ell}$ and width $\frac{n}{2^{\ell-1}}$ and 
  area of $\Theta\left(2^{\gamma i} n   \right)$.
\end{lemma}
\begin{proof}
  Assume we take a level $\ell = 4t(j-1)+r_i$ and let $v$ be a vertex of level $\ell$.
  As discussed, 
  we divide $R(v)$ into $\frac{n}{2^{\gamma i + \ell}}$ sub-rectangles, so
  each sub-rectangle is thus a rectangle of 
  height $2^{\gamma i + \ell}$ and width $\frac{n}{2^{\ell-1}}$ with an area
  of $2^{\gamma i + \ell} \cdot \frac{n}{2^{\ell}-1} = \Theta\left(2^{\gamma i} n   \right)$.

  As the number of vertices of level $\ell$ is $2^{\ell}$, the total
  number of values that we insert at this level is $\frac{n}{2^{\gamma i}}$; this
  only depends on $i$, the epoch and thus, over all the
  $b$ odd bands, we have $\frac{bn}{2^{\gamma i}}$ values inserted.
\end{proof}

For the ease of presentation, we now define a notion \mdef{epoch of change (EoC)}. 
Consider epoch $i$. 
The sub-rectangles that we insert during this epoch are said to have EoC $i$.
Any memory cell that is updated during this epoch has its EoC sets to $i$.
Note that if a memory cell gets updated during epochs $i_1 < \dots< i_\ell$, then it will
have EoC $i_\ell$ at the end.
We use $\eoc(\cv)$ to denote the EoC of a memory cell $\cv$.
%Additionally, we say a cell $\cv$ is younger (resp. older) than 
%a cell $\cu$ if $\eoc(\cv) > \eoc(\cu)$ (resp.  $\eoc(\cv) < \eoc(\cu)$).

\section{Persistent Structures}\label{sec:pers}
\subsection{The Existence Proof.}
In this section, we consider the data structure after we have made all the insertions, at the end of epoch $t$.
Here, we choose to view the fractional cascading problem as the rectangle stabbing problem,
as outlined by our geometric view, where a query is represented by a point $q \in \Q$.
Observe that every query point $q$ is contained inside $tb = \Theta(\log n)$ sub-rectangles
and thus $tb$ is the output size.
Our ultimate goal is to show that the worst-case query time of the data structure is
at least $\alpha tb$.

Let $T(q)$ be a smallest directed tree explored at the query time that outputs all the 
sub-rectangles containing $q$.
Thus, $T(q)$ contains $tb$ cells that produce output; we call these, \mdef{output cells}.
We call a subtree $\C$ of $T(q)$ that has $x$ output cells marked as output cells a \mdef{$x$-connector}
and its size is its number of vertices.
Note that two connectors that have identical set of memory cells and edges are considered to be
different if they have different subsets of cells marked as output.
%Sometimes we refer to $\C$ as just a connector.
If there exists a query $q$ such that for every tree $T(q)$ with $tb$ output cells we have
$|T(q)| \ge \alpha tb$, then, the worst-case query time of the data structure is at least
$\alpha tb$ and thus we have nothing left to prove. 
As a result, in the rest of this proof we assume that for every qurey $q$, there exists one
tree $T(q)$ of size less than $\alpha tb$.

The concept of ``connectors'' is a crucial part of static pointer machine lower bounds.
However, this alone will not be sufficient for a dynamic lower bound. The problem is that we do not
have any control on how these connectors are created.
Ideally, we would like to be able to locate the same connector throughout the entire sequence of updates.
To do that, we define the notions of \mdef{strange} and \mdef{ordinary} connectors. 
A connector is called strange if it contains two nodes $u$ and $w$ such that $w$ is an ancestor of $u$
(in the connector) and $u$ is an output node that outputs a sub-rectangle of EoC $i$ but $w$ is updated
during a later epoch, meaning, $\eoc(w) > i$. 
See Figure~\ref{fig:ord-1}.
\begin{figure}[h]
    \centering
    \includegraphics[scale=0.7]{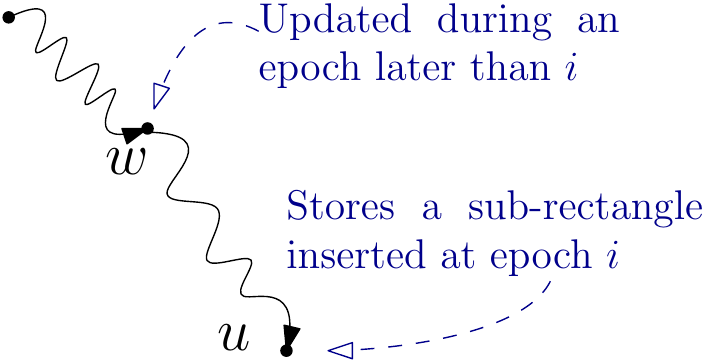}
    \caption{A strange connector.}
    \label{fig:ord-1}
\end{figure}
In plain English, the connector produces an output from an epoch $i$ using a memory cell that was updated
in a later epoch $i'$.
A connector that is not strange is called ordinary. 

Ordinary connectors will be extremely useful later but at this point its not clear if such connectors even exist.
In fact, for large values of $k$, we can clearly see that ordinary $k$-connectors may not exist:
consider $tb$-connectors. Every query outputs exactly $tb$ sub-rectangles, so $T(q)$ is the only $tb$-connector in $T(q)$.
However, $T(q)$ always contains the root of the data structure, meaning,
if the root is updated at every epoch, then there are no ordinary $tb$-connectors.

The main result of this section is  that for an appropriate choice of $k$, ordinary $k$-connectors in fact
exist. 
What we will prove is in fact slightly more complicated; we will show that for most queries $q$, ordinary $k$-connectors
exist.
To do that, we mark certain regions of $\Q$ as \mdef{strange} and we will show that as long as $q$ does not come
from the strange region, then $T(q)$ will contain an ordinary $k$-connector.

\subsection{Strange Regions.}\label{subsec:strange}
In this subsection, we only consider the status of the data structure at the end of epoch $t$. 
Let $\mem$ be this memory structure.

\subparagraph{Marking regions.}
This process uses a parameter $x$, to be fixed later.
To mark the strange region, we preform the following for every integer $z \le x$ and every memory cell
$\cw \in \mem$.
Consider all the trees that can be made, starting from $\cw$ as their root and by following $z$ edges in $\mem$.
Let $T$ be one such tree. 
Consider all possible ways of marking the cells of $T$ as output cells. 
Let $\cu_1, \dots, \cu_r$ be a paritcular choice of output cells in $T$. 
If each $\cu_\ell$, $1 \le \ell \le r$ stores a sub-rectangle $B_\ell$ such that $\eoc(\cw)  > \eoc(B_\ell)$
and if 
%Note that by the definition of a strange connector, if both $\cw$ and $\cu_\ell$ exist in a connector then it would
%be strange. 
$r\ge \frac{z}{\beta}$, for a parameter $\beta$, we mark 
$B_1 \cap \dots \cap B_r$ as strange.

Our main lemma in this section is the following. 
\begin{lemma}\label{lem:strange}
    If $ x \le \frac{t\gamma}{18\beta}$, then the total area of the strange regions is $o(n^2)$.
\end{lemma}
\begin{proof}
  The number of insertions performed at epoch $i$ is 
  $\frac{bn}{2^{\gamma i}}$ by \refl{inssize} and since we are working in the 
  EWIA model, the number of memory cells of EoC $i$ is at most
  \begin{align}
    \frac{bnU}{2^{\gamma i}}. \label{eq:insNum}
  \end{align}
  Pick one memory cell $\cw$ with $\eoc(\cw) = i$.
  Starting from $\cw$ and using $z$ edges, we can form at most $2^{2z}$ subtrees $T$.
  We can also mark the cells in $T$ as output or not in at most $2^z$ different ways. 
  Thus, the total number of marked trees considered is 
  \begin{align}
    \frac{bnU2^{3z}}{2^{\gamma i}}. \label{eq:TNum}
  \end{align}
 Following the definition of our marking process, let $T$ be one such tree and 
let $\cu_1, \dots, \cu_r$ be all the marked cells in $T$.
Assume that 
each $\cu_\ell$, $1 \le \ell \le r$, stores a sub-rectangle $B_\ell$ such that $\eoc(\cw)  > \eoc(B_\ell)$,
as otherwise, no region is marked strange. 
  Assume $B_\ell$ is inserted in the catalog of a node $v_\ell$ at level $k_\ell$.
  Observe that if $B_1 \cap \dots \cap B_r = \emptyset$ then we do not increase the strange region.
  Thus, in the rest of this proof we only consider the case when 
  $B_1 \cap \dots B_r \not = \emptyset$.
  By Observation~\ref{ob:ancestor}, this implies that all the nodes $v_1, \dots, v_\ell$ must
  lie on the same root to leaf path.
  Consequently, this implies all $k_\ell$'s are distinct.
  Assume $B_\ell$ was inserted at epoch $i_\ell < i$.
  This means that $B_\ell$ is a rectangle with width $\frac{n}{2^{k_\ell}-1}$ and with height
  $2^{k_\ell+ \gamma i_\ell}$ by Lemma~\ref{lem:inssize}.

  W.l.o.g, let $B_1$ be the rectangle with minimum width and $B_2$ be the rectangle with the
  minimum height, i.e.,
  $k_1 = \max\left\{ k_1, \dots, k_r \right\}$ and $k_2 + \gamma i_2 = \min\left\{ k_1+ \gamma i_1, \dots, k_r + \gamma i_r \right\}$.
  We have, $B_1 \cap \dots \cap B_r  \subset B_1 \cap B_2$.
  Furthermore, 
  \begin{align}
    \area(B_1 \cap B_2) = \frac{n}{2^{k_1-1}} 2^{k_2 + \gamma i_2} = \frac{2n}{2^{k_1-k_2 + \gamma(i-i_2)}} 2^{\gamma i}.\label{eq:area}
  \end{align}
  With a slight abuse of the notation and to reduce using extra variables, 
  define $k_3 = \min\left\{ k_1, \dots, k_r \right\}$ and $i_4 = \min\left\{ i_1, \dots, i_r \right\}$ and note that the subscripts 
  $1, 2, 3, 4$ could refer to distinct or same sub-rectangles. 
  Remember that we have $i_1, \dots, i_r < i$ and thus $i > i_3$.
  Also we have $k_3+i_3 \gamma \ge k_2 + i_2 \gamma$ and thus
  \begin{align}
    k_3 + i \gamma &> k_3 + i_3 \gamma \ge k_2 + i_2 \gamma \Longrightarrow \nonumber\\
    k_3 &\ge k_2 + i_2 \gamma - i \gamma  \label{eq:k3}
  \end{align}
  and also since $k_1 \ge k_4$ and $k_2 + \gamma i_2 \le k_4 + \gamma i_4$,
  \begin{align}
    k_1 + i_4 &\gamma \ge k_4 + i_4 \gamma \ge k_2 + \gamma i_2 \Longrightarrow \nonumber \\
    i_4 &\ge \frac{k_2 - k_1}{\gamma} + i_2. \label{eq:i4}
  \end{align}
  Remember that by our definitions, for every $1 \le \ell \le r$, we have
  $k_3 \le k_\ell$ and $i_4 \le i_\ell < i$.
  Furthermore, as argued previously, the sub-rectangles $B_1, \dots, B_r$ have been inserted in
  nodes of $v_1, \dots, v_r$ that all lie on the same root to leaf path in $\cat$.
  Additionally, if for two indices $\ell$ and $\ell'$ we have $i_{\ell} = i_{\ell'}$, then it follows
  that the sub-rectangles $B_\ell$ and $B_{\ell'}$ were inserted in the same epoch.
  But we only insert sub-rectangles into odd bands and between each two odd bands, there are at least $2t$ levels.
  Thus, $|k_{\ell} - k_{\ell'}| \ge 2t$.
  Putting this all together, it follows that for each value of $i_\ell$, $i_4 \le i_\ell < i$, there can 
  be at most $\frac{k_1 - k_3}{2t}$ sub-rectangles in the list $B_1, \dots, B_r$.
  As a result, the maximum number of sub-rectanles in the list is $\frac{(i-i_4)(k_1 - k_3)}{2t}$, or in other words, 
  \begin{align}
    r \le \frac{(i-i_4)(k_1 - k_3)}{2t}.\label{eq:r}
  \end{align}
  We now use inequalities (\ref{eq:k3}) and (\ref{eq:i4}) to obtain that
  \begin{align}
    \begin{aligned}
    r &\le \frac{(i-\frac{k_2 - k_1}{\gamma} - i_2) (k_1 - k_2 - i_2 \gamma + i \gamma )}{2t} = \nonumber \\
      & \frac{(i\gamma -k_2 + k_1 - i_2 \gamma) (k_1 - k_2 - i_2 \gamma + i \gamma )}{2t\gamma}  = \frac{(\gamma(i-i_2)+ k_1 -k_2 )^2}{2t\gamma} 
    \end{aligned}
  \end{align}
  Also remember that we had $\frac{z}{\beta} <r$.
  Thus,
  \begin{align} 
      z < \frac{\beta(\gamma(i-i_2)+ k_1 -k_2 )^2}{2t\gamma}. \label{eq:z}
  \end{align}
  We now consider two cases. 
  
  \subparagraph{Case (i) $z \ge \gamma/4$.}
  As $z \le x \le \frac{t\gamma}{18\beta}$, by multiplying both sides with \refe{z} we get
  \begin{align} 
      z^2 < \frac{(\gamma(i-i_2)+ k_1 -k_2 )^2}{36} \Longrightarrow 6z < \gamma(i-i_2)+ k_1 -k_2.   \label{eq:z2}
  \end{align}
  Observe that we can plug in this inequality into \refe{area} and obtain that
  \begin{align}
      \area(B_1 \cap B_2) <  \frac{2n}{2^{6z}} 2^{\gamma i}.\label{eq:area2}
  \end{align}
  We now calculate how much the strange region expands in this case.
  We need to sum over all choices of $z$ ($x$ choices) and in each sum, we need to multiply 
  the number of trees we consider given in \refe{TNum} and the area 
  of increase given in \refe{area2}.
  So the increase in the area of the strange region is bounded by
  \begin{align}
      \sum_{z=1}^x  \frac{bnU2^{3z}}{2^{\gamma i}}\cdot \frac{2n}{2^{6z}} 2^{\gamma i} <\sum_{z=1}^x \frac{2bn^2 U}{2^z} \le \nonumber \\
      \frac{2bn^2 Ux}{2^{\gamma/4}}  \le \frac{2bn^2 U\gamma \log n}{2^{\gamma/4}} \le  \frac{n^2}{\log^2n } \label{eq:strange}
  \end{align}
  where the first inequality uses the assumption that $z \ge \gamma/4$ and the second
  inequality uses the fact that $x \le \frac{t \gamma}{18\beta} < t \gamma \le \gamma\log n $ and 
  the third inequality uses the observation that 
  $\gamma \ge 4 (2+ \log b+ \log U + \log \gamma +  3 \log\log n)$;
  the observation can be verified for large enough $n$ if we plug in the value $\gamma = 100(\log n\log n + \log U)$.
  Since we have less than $\log n$ epochs, the total increase of the area of the strange region is bounded by
  $\frac{n^2}{\log n}$, in this case.

  \subparagraph{Case (ii) $z < \gamma/4$.}
  Here, we simply observe that 
  \begin{align}
      \area(B_1 \cap B_2) = \frac{2n}{2^{k_1-k_2 + \gamma(i-i_2)}} 2^{\gamma i}< \frac{2n}{2^{\gamma}} 2^{\gamma i}\label{eq:area3}
  \end{align}
  since $i_2 < i$ and thus $\gamma (i-i_2) \ge \gamma$.
  As before, we can compute the total increase in the area of the strange region by  
  summing over all choices of $z$ (less than $\gamma$ choices) and in each sum, we need to multiply 
  the number of trees we consider given in \refe{TNum} and the area 
  of increase given in \refe{area3}.
  We obtain that increase is bounded by
  \begin{align}
      \sum_{z=1}^{\gamma/4}  \frac{bnU2^{3z}}{2^{\gamma i}}\cdot \frac{2n}{2^{\gamma}} 2^{\gamma i} =\sum_{z=1}^{\gamma/4} \frac{2bn^2 U2^{3z}}{2^\gamma} \le  \nonumber \\
      \sum_{z=1}^{\gamma/4}\frac{2bn^2 U}{2^{\gamma/4}} = \frac{\gamma b n^2 U }{2\cdot 2^{\gamma/4}} < \frac{n^2}{\log^2n} \label{eq:strange}
  \end{align}
  where the last inequality 
  uses  that $\gamma = 100(\log n\log n + \log U)$.
  As before, the total increase in the area of the strange region is bounded by $\frac{n^2}{\log n}$.
\end{proof}

\begin{lemma}\label{lem:ordcon}
    Let $q$ be a query point from the ordinary region of $\Q$ and let $T(q)$ be the tree traversed at the
    query time such that $|T(q)| < \alpha bt$. 
    Let  $c$ be a parameter such that $c \le \frac{t\gamma}{18\beta^2}$ where $\beta$ is the parameter used in Section~\ref{subsec:strange}.
    We can find a number, $m$, of disjoint ordinary connectors such that the $i$-th connector 
    has $k_i$ output cells and has size at most $\beta k_i$ with $c/4 \le k_i \le c$.
    Furthermore the connectors contain at least $tb(1 - 5\frac{\alpha}{\beta}) - \frac{c}{2}$
    of the output cells.
\end{lemma}
\begin{proof}
  We first find  a set of (initial) connectors iteratively, by cutting off subtrees of $T(q)$.
  Then, for each initial connector, we find a subset of it that is ordinary. 

  We find the initial set iteratively.
  In the $i$-th iteration, if there are fewer than $c/2$ output cells left, then we are done.
  Otherwise, consider the lowest cell $\cv$ that has at least $c/2$ output cells in its subtree;
  as $T(q)$ is a binary tree, $\cv$ has at most $c$ in its subtree as otherwise, one of its children will have 
  at least $c/2$ output cells and will be lower than $\cv$.
  Let $k_i$ be the number of output cells at subtree of $\cv$ and $s_i$ be its size. 
  If $s_i < \beta k_i/4$, then we have found one initial connector,
  we add it to the list of initial connectors we have found and delete $\cv$ and its subtree
  and continue with the next iteration. 
  However, if the subtree at $\cv$ is larger than $\beta k_i/4$, 
  we call its output cells \mdef{wasted} and again we delete the subtree hanging at $\cv$ and then continue.
  Let $K$ be the total number of wasted cells and with a slight abuse of notation, let 
  $\C_1, \dots, \C_m$ be the set of initial connectors that we have found where
  $\C_i$ has $k_i$ output cells and has size $\alpha_i k_i$.
  Observe that $K < \frac{4\alpha tb}{\beta}$ since each wasted cell can be charged to at least $\beta/4$ other cells
  in the tree $T(q)$ and there are at moest $\alpha bt$ cells in $T(q)$.
  Every output cell is either wasted, placed in a connector or it is among the fewer than $c/2$ output cells left
  at the end of our iterations.
  Thus,
  \[
    \sum_{i=1}^m k_i \ge tb -  K -c/2 \ge tb - \frac{4\alpha tb}{\beta} - c/2.
  \]

    Observe that the value $c$ is chosen such that the size of each initial connector is at most
    $\beta c/4 < \frac{t\gamma}{18\beta}$, satisfying the condition of Lemma~\ref{lem:strange}.
    Consider one connector $\C_i$.
    If $\C_i$ is ordinary, then we are done. 
    Otherwise, we will ``unmark'' some of its output cells, meaning, we will obtain another
    connector $\C'_i$ that has the exact same set of edges and vertices as $\C_i$ but whose set of output cells is a subset of the output
    cells of $\C_i$.
    This unmarking will make sure that $\C'_i$ is ordinary and furthermore, it still contains 
    ``almost'' all of the output cells of $\C_i$.
    
    We now focus on $\C_i$.
    Consider a cell $\cw \in \C_i$.
    Let $\Delta(\cw)$ be the set of cells in the tree $\C_i$ that is hanging off $\cw$ that \mdef{conflict} with $\cw$, meaning,
    for each cell $\cu \in \Delta(\cw)$, $\eoc(\cw) > \eoc(\cu)$.
    We unmark all the elements of $\Delta(\cw)$, for every cell $\cw \in \C_i$.
    The main challenge is to show that at least a fraction of the output cells of $\C_i$ survive this unmarking operation. 
    Nonetheless, it is obvious that $\C'_i$ is ordinary (note that $\C'_i$ could have no output cells, in which case it is still
    ordinary).

    Let $\Delta$ be the set of cells in $\C_i$ that have conflicts, i.e., $\Delta$ includes all the cells $\cw$ such that
    $\Delta(\cw) \not = \emptyset$. 

    Consider two nodes $\cw_1$ and $\cw_2$ such that $\cw_1$ is an ancestor of $\cw_2$. 
    Observe that if $ \eoc(\cw_1) \ge \eoc(\cw_2) $, then 
    $\Delta(\cw_2) \subset \Delta(\cw_1)$: consider a cell $\cu \in \Delta(\cw_2)$.
    By definition, $\cu$ is in the subtree of $\cw_2$ and thus also in the subtree of $\cw_1$ and furthermore,
    we have $\eoc(\cw_2) > \eoc(\cu)$ but since $\eoc(\cw_1) \ge \eoc(\cw_2)$,  we also have
    $\eoc(\cw_1) > \eoc(\cu)$ and thus $\cu \in \Delta(\cw_1)$.
    This means that we can essentially ignore the unmarking operation on $\Delta(\cw_2)$ as
    $\cw_1$ will take care of those conflicts. 
    We now ``clean up'' $\Delta$ by removing any such cells $\cw_2$ from $\Delta$. 
    
    As a consequence, we can now assume that $\Delta$ has the following property:
    for every two cells $\cw_1$ and $\cw_2$ in $\Delta$ where $\cw_1$ is an ancestor of $\cw_2$,
    we have $\eoc(\cw_1) < \eoc(\cw_2)$.

    We now consider an iterative unmarking operation.
    Let $\cw$ be the memory cell in $\Delta$ that is closest to the root of $\C_i$.
    Let $\cw_1, \dots, \cw_r$ be the cells in $\Delta$ that belong to the subtree of $\cw$.
    Since we have cleaned up $\Delta$, it follows that 
    $\eoc(\cw) < \eoc(\cw_1), \dots, \eoc(\cw_r)$.
    Let $\Delta'(\cw)= \Delta(\cw) \setminus(\Delta(\cw_1) \cup \dots \Delta(\cw_r))$.
    We now make two claims. 

    \subparagraph{Claim (i).}
    For every cell $\cu \in \Delta'(\cw)$, $\cu$ is not contained in the subtrees of $\cw_1, \dots, \cw_r$:
    this claim is easily verified since if $\cu$ is in the subtree $\cw_j$ for some $j$, then we must have
    $\eoc(\cu) < \eoc(\cw) < \eoc(\cw_j)$ and thus $\cw$ also belongs to $\Delta(\cw_j)$ which is a contradiction
    with the assumption that $\cu \in \Delta'(\cw)$.

    \subparagraph{Claim (ii).}
    Let $\U(\cw)$ be the set of cells that is the union of all the paths that connect
    $\cu$ to $\cw$ for all $\cu \in \Delta'(\cw)$.
    We claim $|\U(\cw)| > \beta |\Delta'(\cw)|$:
    This claim follows because of the way the strange regions are marked in Subsection~\ref{subsec:strange};  
    we have $|\U(\cw)| \le |\C_i| \le\frac{t\gamma}{18\beta}$ and thus $|\U(\cw)|$ is less than or equal than
    $x$, the parameter we used in marking the strange region.
    This means that, at some point during the process of marking strange regions,
    we would consider the cell $\cw$ and the tree $T$ with root $\cw$ that is formed by
    exactly the of edges in $\U(\cw)$. Also, at some point we consider exactly the set $\Delta'(\cw)$ 
    as the marked cells of $T$.
    Now, if $|\U(\cw) \le \beta |\Delta'(\cw)|$
    then we would mark the intersection of all the sub-rectangles of $\Delta'(\cw)$ as strange which in return
    implies that $q$ is inside the strange region, a contradiction. 
    Thus,  $|\U(\cw)| > \beta |\Delta'(\cw)|$.
    We now unmark the cells $\cu_1, \dots, \cu_r$ and charge this operation to $\U(\cw)$; by what we have just said,
    every edge in $\U(\cw)$ receives at most $\frac{1}\beta$ charge.

    Now observe that we can remove $\cw$ from $\Delta$. 
    All the other cells that have conflict with $\cw$ will be unmarked by other cells in $\Delta$, by the definition of 
    $\Delta'(\cw)$.
    We continue the iterative unmarking operation by taking the next highest cell $\Delta$ and so on.
    This operation unmarks all the conflicting cells and thus leaves us with an ordinary connector.
    Furthermore, each unmarking operation, charges some edges of $\C_i$ with $\frac{1}{\beta}$ charge.
    We claim each edge receives at most one charge and to do that, we simply observe that
    for two cells $\cw$ and  $\cw'$, the sets $\Delta'(\cw)$ and $\Delta'(\cw')$ are disjoint. 
    Assume w.l.o.g, that $\cw$ was selected first and that $\cw_1, \dots, \cw_r$ were the highest cells in $\Delta$ in the subtree
    of $\cw$. Now observe that $\cw'$ must either lie outside the subtee of $\cw$ or it must lie in the subtree of one of the
    cells $\cw_1, \dots, \cw_r$, revealing that $\Delta'(\cw)$ and $\Delta'(\cw')$ are disjoint. 

    Thus, if $f_i$ is the total number of cells that get unmarked by this operation, we must have
    $f_i \beta < |\C_i|$.
    Thus, at most $\frac{|\C_i|}{\beta}$ output cells are unmarked to obtain an ordinary connector
    of size
    at least 
    \[ 
      k_i - \frac{|\C_i|}{\beta} \ge k_i - \frac{\beta k_i/4}{\beta}  \ge 3k_i/4 > c/4.
    \]
    In addition, over all the indices $i, 1 \le i \le m$, 
    the total number of output cells that are unmarked is
    \[
      \sum_{i=1}^m \frac{|\C_i|}{\beta} \le \frac{\alpha tb}{\beta}.
    \]
    Now the lemma follows since every output cell is either wasted, unmarked or is among the $c/2$ left over cells.
\end{proof}

\subsection{The Structure of an Ordinary Connector.}

Let $\mem$ be the memory graph  at the end of epoch $t$.
Let us review what Lemma~\ref{lem:ordcon} gets us. 
First, observe that by the choice of our parameters, we have $c= \frac{t \gamma}{18 \beta^2} > 8t$ 
(by having $\varepsilon$ in the definition of $\beta$ small enough).
Second, recall that the output size of any query is $tb$ and furthermore, we have assumed that every 
query can be answered by traversing a subtree $T(q)$ of $\mem$ of size at most $\alpha tb$.
By Lemma~\ref{lem:ordcon}, we can set $c=8t$ and decompose $T(q)$ 
into ordinary connectors of size between $2t$ and $8t$.
Crucially, in doing so, we only need to leave out at most $5tb\alpha/\beta + 4t$ output elements. 
As a result, for every query point $q$ that is not in the strange region, 
there exists a ordinary $k$-connector $\C$ of size at most $\beta k$, with 
$2t \le k \le 8t$, whose output sub-rectangles contain $q$ and at least one of its
sub-rectangles has been inserted no later than epoch $\frac{5tb\alpha/\beta + 2t}{b} \le \frac{6t \alpha}{\beta}$.
Pick one such connector $\C$ and call it the \idef{main ordinary connector of $q$}.
Let $i_0$ be the earliest epoch from which $\C$ has an output cell. 
To simplify the analysis later, we assume $\C$ has exactly $k_0$ output cells 
from epoch $i_0$; for now, we treat the integers $i_0$ and $k_0$ fixed, i.e., we only consider ordinary connectors
with exactly $k_0$ output cells from epoch $i_0$.
In addition, from now on, we will only focus on the main ordinary connectors.

%Observe that since we would like to prove a lower bound on the size of $\C$,  we can assume that $\C$ is
%\mdef{minimal}, i.e., all the leafs in $\C$ are output cells; any leaf that is not can be deleted safely without changing
%%the set of output cells of $\C$ and will result in a small connector. 
%Thus, in the rest of this proof, we will assume that $\C$ is minimal.
The fact that $\C$ is ordinary yields very important and
crucial properties that are captured below, but intuitively, it implies that we can decompose $\C$ into
$t$ ``growth spurts'' where during each growth spurt, it ``grows'' a number of
branches. 

Let $\out_i(\C)$, for $1 \le i \le t$, be the set of output cells that store
sub-rectangles that belong to epoch $i$ (i.e., they were inserted during epoch $i$). 
Let $\cv$ be the root of $\C$. 
The younger $\C$ at time $i$, $\young{\C,i}$ is defined as the union of
all the paths that connect cells of $\out_i(\C)$ to $\cv$;
by the ordinariness of $\C$, the subgraph $\young{\C,i}$ existed at the end of epoch $i$ and it has not been updated since then;
this structure has been preserved throughout the later updates.
The $i$-th growth spurt, $\spurt{\C,i}$, is defined as $\young{\C,i} \setminus \young{\C,i-1}$ where $\young{\C,0}$ is an empty set.
Observe that $\spurt{\C,i}$ is a forest and we call each tree in it a \idef{branch}.
We call the edge that connects a branch to $\young{\C,i-1}$ a \idef{bridge}.
For a cell  $\cu$ in a branch,  we define the \idef{connecting path} of $\cu$ to be the path in 
$\C$ that connects $\cu$ to its bridge. 
If a branch has at least $\mu$ edges, for a parameter $\mu$, we call it a \idef{large branch}
otherwise it is a \idef{small branch}. 
For every output cell $\cu$ in $\spurt{\C,i}$, consider the connecting path $\pi$ of $\cu$ that connects
$\cu$ to a cell $\cu'$ adjacent to its bridge. 
We define the \mdef{bud} of $\cu$ to be the first cell on path $\pi$ (in the direction of $\cu'$ to $\cu$)
which is updated in epoch $i$. 
Observe that since $\C$ is ordinary, bud of $\cu$ exists but  it is possible for an output cell 
to be its own bud. 
See Figure~\ref{fig:ord-struc}.

\begin{figure}[h]
  \centering
  \includegraphics[scale=0.6]{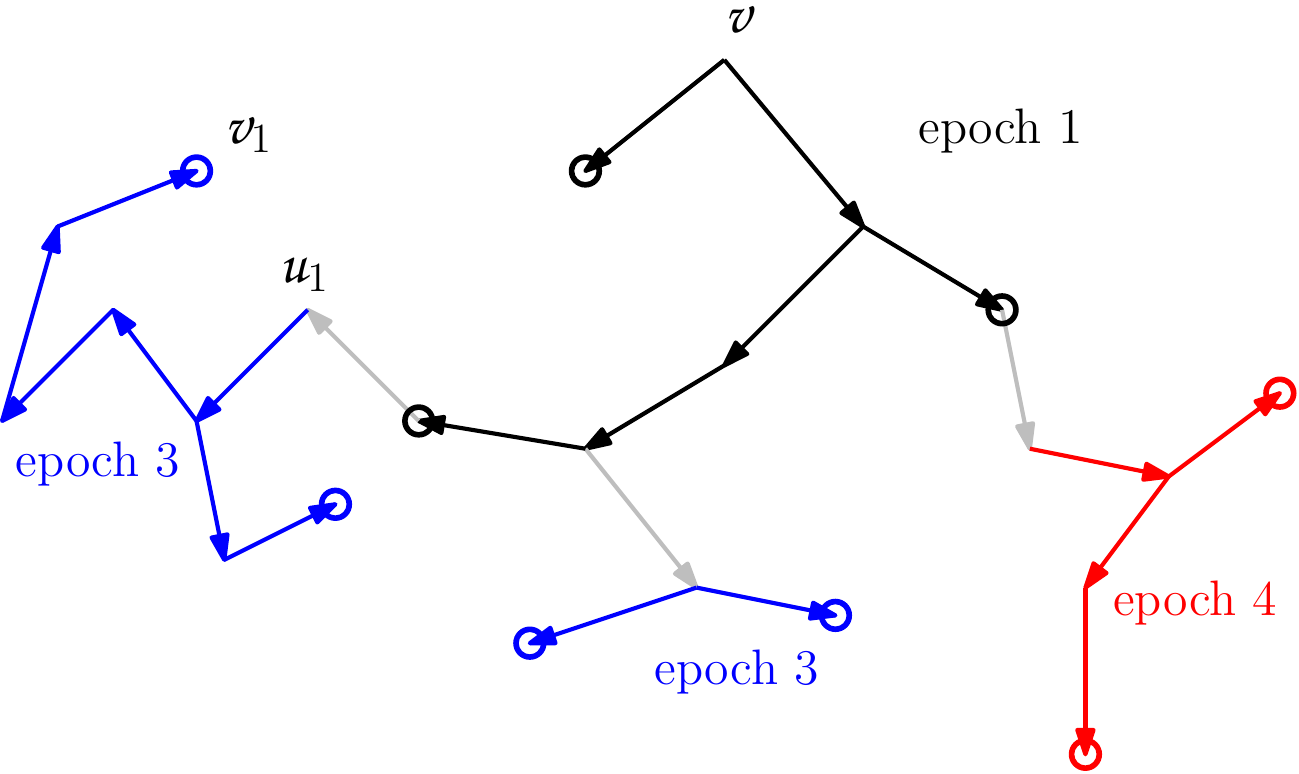}
  \caption{The structure of an ordinary connector. Output cells are marked with large circles. Black, blue, and red output cells are those
    that belong to epochs 1, 3, and 4 respectively. 
    Bridges are grey edges. 
    The black subtree is $\young{\C,1}$. The blue and blue subtrees together with the two connecting bridges form $\young{\C,3}$.
    $\spurt{\C,3}$ is the forest containing the two blue subtrees.
    The connecting path of $\cv_1$, is the path that connects $\cu_1$ to $\cv_1$.
}
  \label{fig:ord-struc}
\end{figure}

\subsubsection{Encoding an Ordinary Connector.}\label{subsec:enc}
We now try to represent the main ordinary connector $\C$ of any query $q$, 
together with most of the information given above, as a bit-vector, $\enc{\C}$.
We assume that every cell has two out-going edges (by adding dummy edges) and one of them
is labelled as the ``left edge'' while the other one is marked as the ``right edge''.
We assume this marking does not change unless the algorithm updates the memory cell.
For a given ordinary connector $\C$ with root $\cv$, the bit-vector $\enc{\C}$ encodes the following:
(i) the shape of $\C$, i.e., starting from $\cv$ and in DFS ordering of $\C$, for every cell $\cu \in \C$, we use two bits to encode 
whether the left or right edge  going out of $\cu$ belongs to $\C$ 
(ii) the number of output cells from each epoch, 
(iii) for every edge we encode whether it's a bridge or not,
(iv) for every cell in $\C$ we encode whether it's an output cell, and/or a bud
(v) we encode which branches belong to epoch $i_0$, and finally,
(vi) for every large branch we encode the epoch it belongs to.
Note that $\enc{\C}$ omits some major information: for output cells that belong to small branches, we do not
encode the epoch they belong to.
Also, note that from (i) and (vi), we can deduce the exact number, $k_i$, of the output cells of epoch $i$ that
are on the small branches. 

\begin{lemma}\label{lem:enc}
  Let $\C$ be an ordinary $O(t)$-connector of size $O(t\beta)$.
  Then, at most $O(t\beta)$ bits are required in $\enc{\C}$.
\end{lemma}
\begin{proof}
  Let $z = |\C|$.
  Encoding (i) requires $2z$ bits, for (iii), (iv), and (v) we require at most $4$ bits  per cell,
  for a total of $6z$ bits.
  Note that a partition of $\C$ into branches is uniquely determined by the set of bridges.

  Now consider large branches; as each large branch contains at least $\mu$ edges, the total number
  of large branches is at most $z/\mu$.
  Encoding the epoch requires $\log t \le \log\log n$ bits and thus (vi) requires at most 
  $z\log\log n / \mu = 8z$ bits (recall that $\mu=\frac{\log \log n}{8}$).

  Finally, (ii) requires 
  $\log {\left({O(t)\choose t}  \right)} = O(t)$ bits, since it is equivalent to the problem of distributing
  $O(t)$ identical balls into $t$ distinct bins.
  Thus, in total we need $O(z) = O(t \beta)$ bits. 
\end{proof}

\section{The Lower Bound Proof}\label{sec:lb}

\subsection{Definitions.}
We use the notation $\mem^i$ to denote the memory graph at the end of epoch $i$. 
$\mem^0$ denotes the memory graph at the beginning of the first epoch. 
Consider an epoch $i$.
Consider a string $S$ that is the valid encoding of an ordinary connector.
Recall that at all times, for every cell its two out-going edges are labelled as left and right.
This implies that for every cell $\cv \in \mem^i$, the string $S$ 
uniquely identifies a connected subgraph formed by taking exactly
the edges encoded by $S$.
We denote this by $T_S(\cv,\mem^i)$ and note that this maybe not be a tree or even a simple graph.

During each epoch, a region of $\Q$ is marked as \idef{forbidden}.
We always add a region to the forbidden region, never remove anything from the  forbidden region.
Another way of saying this is that the forbidden region always expands, i.e., the forbidden region at the end of
epoch $i$ will contain the forbidden region at the end of epoch $i-1$.
At the end of epoch $t$, the query that we will
ask will not be in the forbidden region (also it will not be in the strange region).
We will keep the invariant that the forbidden region is a small fraction of $\Q$.
In addition, for a region $X \subset \Q$ and at an epoch $i$, we may  
\idef{conditionally forbid (CF) $X$ at a cell $\cv$}.
This  implies the following. 
If at any later epoch $i'> i$, the cell $\cv$ or any cell in the $\mu$-in-neighborhood of $\cv$ 
is updated by the algorithm, then the region $X$ is added to the forbidden region. 
For ease of description, we sometimes say that $X$ is CF at the beginning of
epoch $i$, which is equivalent to $X$ being CF at the end of epoch $i-1$.

%The concept of conditionally forbidding region allows us to have a much l

We now describe the forbidden regions. 

\subsubsection{The First Type of Forbidden Regions.}
Recall the details of our construction: at each epoch, we take $b$ levels, one level in each odd band.
Depending on which levels we take, we insert a different sequence of sub-rectangles in the
data structure. 
\deff{f1}{
  Consider the beginning of epoch $i$ and $\mem^{i-1}$.
  For each vertex $\cv \in \mem^{i-1}$ we assign to $\cv$ a set of sub-rectangles as \idef{tags} as follows.
  Assume there is  a choice of levels to take in epoch $i$ that results in an insertion sequence that
  leads to the following:
  at the end of epoch $i$, the algorithm has updated the cell $\cv$ (its contents or pointers) 
  and it has stored a sub-rectangle $\B$ that was inserted in epoch $i$ at a memory cell in the $\mu$-out-neighborhood of $\cv$.
  Then we add $\B$ as a \mdef{tag} to the set of tags of $\cv$.

  The set of tags of a cell $\cv$ is thus determined at the beginning of epoch $i$.
  At the beginning of epoch $i$, we consider the (geometric) arrangement $\A_\cv$ of the tags of $\cv$.
  We add any point that is contained in more than $10b \log t$ tags to the forbidden region. 
}

\begin{lemma}
  The area of the first type of forbidden regions is increased by $o(n^2/\log n)$ at every epoch. 
\end{lemma}
\begin{proof}
  Recall the details of our construction and remember that at each epoch $i$, we take $b$ 
  levels among the untaken levels of each odd band.
  By Observation~\ref{ob:choices}, 
  in each odd band we have at most $2t$ choices.
  Over $b$ odd bands this adds up to at most $(2t)^{b}$ choices for which levels to take. 
  As the size of $\mu$-out-neighborhood of $\cv$ is at most $2^{\mu}$, 
  the number of sub-rectangles added to $\cv$ as tags is at most 
  $2^{\mu} (2t)^b$.
  The number of geometric cells in the arrangement $\A_\cv$ is at most $O\left( \left(
  2^{\mu} (2t)^b \right)^2 \right) = O(2^{2\mu} (2t)^{2b})$.

  Consider a geometric cell $\Delta$ that is contained in $r$ sub-rectangles $\B_1, \dots, \B_r$ in $\A_\cv$.
  By Observation~\ref{ob:ancestor},  for $\B_1, \dots, \B_r$ to have non-empty
  intersection, they all must belong to the
  same root to leaf path in $\cat$, i.e., $\B_1, \dots, \B_r$ must have different levels. 
  Thus, w.l.o.g, assume the sub-rectangles $\B_1, \dots, \B_r$ are sorted by the increasing value of their levels. 
  As a result, if $\ell_1$ is the level of $\B_1$ and $\ell_r$ is the level of $\B_r$, 
  we must have $\ell_r - \ell_1 \ge r$.
  By Lemma~\ref{lem:inssize}, $\B_1$ is a rectangle of 
  height $2^{\gamma i + \ell_1}$ and width $\frac{n}{2^{\ell_1-1}}$ and $\B_r$ is a rectangle of  
  height $2^{\gamma i + \ell_r}$ and width $\frac{n}{2^{\ell_r-1}}$.
  As a result,
  \[
    \area(\B_1 \cap \B_r) \le 2^{\gamma i + \ell_1}\cdot \frac{n}{2^{\ell_r-1}} \le \frac{n 2^{\gamma i}}{2^{\ell_r - \ell_1 -1}} \le  \frac{n 2^{\gamma i}}{2^{r-1}}.
  \]
  Thus, if we consider the arrangement $\A_\cv$, every geometric cell that is contained in $10b\log t$ sub-rectangles will have
  an area of at most $ \frac{n 2^{\gamma i}}{2^{10b\log t-1}}$ by above, meaning, in total $\A_\cv$ can have an area
  of  
  \begin{align}
    \frac{n 2^{\gamma i}}{2^{10b\log t-1}}\cdot O(2^{2\mu} (2t)^{2b}) = O\left( \frac{n2^{\gamma i}2^{2\mu}}{t^{6b}} \right) \label{eq:f2}
  \end{align}
  in such cells.

  Finally, observe that we only consider cells $\cv$ that are updated by the algorithm. 
  By Lemma~\ref{lem:inssize}, during epoch $i$, a total of $\frac{bn}{2^{\gamma i}}$ values are inserted, 
  and since we are operating in the EWIA model, this means that
  a total of $\frac{bn U}{2^{\gamma i}}$ cells can be updated;
  multiplying this by Equation~\ref{eq:f2} reveals that the total increase in the area of the forbidden region is
  bounded by 
  \begin{align*}
    O\left( \frac{n2^{\gamma i}2^{2\mu}}{t^{6b}} \cdot \frac{bnU}{2^{\gamma i}} \right)  = O\left( \frac{n^2b2^{2\mu}U}{t^{6b}} \right) = o(\frac{n^2}{\log n}).
  \end{align*}
\end{proof}

\subsubsection{The Second Type of Forbidden Regions}

\deff{f2}{
  Let $\cv$ be a cell that is updated during epoch $i$. 
  If a cell $\cu$ in the $\mu$-out-neighborhood of $\cv$ stores a sub-rectangle $\B$ that is inserted
  in the previous epochs (i.e., epochs $1$ to $i-1$), then we add $\B$ to the forbidden region. 
}

\begin{lemma}\label{lem:forbidden}
  The area of the second type of forbidden regions is increased by $o(n^2/\log n)$ at every epoch. 
\end{lemma}
\begin{proof}
  Observe that the area of $\B$ is at most $\Theta\left(2^{\gamma (i-1)} n   \right)$ by Lemma~\ref{lem:inssize}.
  However, during epoch $i$, we insert $\frac{bn}{2^{\gamma i}}$ values.
  Since we are operating in the EWIA model, this means that the data structure can modify at most
  $U\frac{bn}{2^{\gamma i}}$ memory cells $\cv$.
  The size of the $\mu$-out-cone of $\cv$ is at most $O(2^{\mu})$ and thus in total
  we can increase the area of the forbidden region by 
  \[
    O\left(2^{\mu}\cdot U\frac{bn}{2^{\gamma i}}\cdot 2^{\gamma (i-1)} n \right) = O\left( \frac{2^{\mu}Ub n^2}{2^\gamma} \right) < \frac{n^2}{\log^2 n}
  \]
  by the definition of $\gamma$.
\end{proof}

\subsubsection{The Third Type of Forbidden Regions}
\deff{f3}{
  Let $\cv$ be a cell that is updated during epoch $i$. 
  If there is a cell $\cu$ in the $\mu$-out-neighborhood of $\cv$, such that 
  a region $X$ is conditionally forbidden at $\cu$, then we add $X$ to the forbidden region. 
}

\begin{lemma}\label{lem:cf}
    If each cell contain a conditionally forbidden region of area at most 
    $\Theta\left(\frac{2^{\gamma i} n}{U \log^4 n}   \right)$, then 
      the area of the third type of forbidden regions is increased by $o(n^2/\log n)$ at every epoch. 
\end{lemma}
\begin{proof}
    The proof is very similar to the previous lemma:
  during epoch $i$, the data structure can modify at most
  $U\frac{bn}{2^{\gamma i}}$ memory cells $\cv$.
  The size of the $\mu$-out-cone of $\cv$ is at most $O(2^{\mu})< \log n$.
  Thus, 
  we can increase the area of the forbidden region by 
  \[
    O\left(2^{\mu}\cdot U\frac{bn}{2^{\gamma i}}\cdot \frac{2^{\gamma i} n}{U\log^4 n}  \right)   < \frac{n^2}{\log^2 n}
  \]
  by the definition of $\gamma$.
\end{proof}

\subsubsection{Living Trees.}
At each epoch, we will maintain a set of \idef{living} trees.
In particular, every memory cell $\cv$ will have a set of living trees with $\cv$ as their root.
Note that this set could be empty.
A living tree is defined with respect to a string $S$ that encodes a (main) ordinary connector and
a root cell $\cv$.
A living tree $T$ (at the end of epoch $i$) must be a subtree of $T_S(\cv,\mem^i)$.
Unlike $T_S(\cv, \mem^i)$, a living tree must be a tree. 
Informally speaking, 
a living tree $T$ is tree that has the potential to ``grow'' into a ordinary connector by the
end of epoch $t$, so, it should not contradict the encoding $S$.

\deff{living}{
More formally, we require  the following:
a living tree $T$ should be a subtree of $T_S(\cv, \mem^i)$, and it must have $\cv$ as its root.
$T$ should be ordinary, and it should be consistent with encoding $S$, both in its structure
and its number of output cells.
To be more specific, $T$
must contain all the long branches belonging to epochs $i$ and before (as encoded by $S$) 
with correct sub-rectangles stored at their output cells (i.e., if a long branch is marked to belong to epoch
$j$, $j\le i$, a sub-rectangle inserted during epoch $j$ must be stored at its output cells).
For every small branch encoded by $S$, $T$ should either fully
contain it or be disjoint from it. 
If $T$ contains a small branch, then a sub-rectangle inserted at an epoch $j$, $j\le i$,
must be stored in its output cells.
Furthermore, $T$ must have exactly $k_j$ output cells on small branches that belong
to epoch $j$, where $k_j$ is the number (implicitly) encoded by $S$, for every $1 \le j \le i$.
In addition, $S$ should be the encoding of a main connector, meaning, for $i<i_0$,
$T$ should have no output cells from epoch $i$ but exactly $k_0$ output cells
from epoch $i_0$.
Finally, for every output cell in $T$, its bud should be at the correct position,
as encoded by $S$.
In this case, we say $S$ super-encodes $T$.
}

We need to remark that by the above definition, there are no living trees before
epoch $i_0$.

\deff{growth}{
    Consider a living tree $T$ at the end of epoch $i$. Remove all the branches that contain
    output cells that belong to epoch $i$. 
    We obtain a living tree $T'$ at the beginning of epoch $i$ and 
  we say that $T$ is \idef{grown} out of $T'$.
  In addition, if we have a sequence of living trees $T_i, T_{i+1}, \dots, T_j$ in epochs
  $i$, $i+1$, \dots, $j$ respectively, where each living tree is grown out of the preceding one,
  then we can also say $T_j$ is grown out of $T_i$ (in epoch $j$).
}

\deff{adjacent}{
  Consider a living tree $T$  with super-encoding $S$ and the connected
  graph $T_S(\cv, \mem^i)$. 
  While $T_S(\cv, \mem^i)$ may not be a tree during epoch $i$, nonetheless,
  $S$ includes the encoding of branches that are connected to $T$ via bridges.
  We call those the \mdef{adjacent branches} of $T$.
  We can use the quantifier small or large as appropriate to refer to those branches.
  Similarly, an adjacent  output cell is one that is located on an adjacent branch. 
  And same holds for an adjacent bud.
}

\subparagraph{Some intuition.}
The following two definitions are a bit technical. 
However, the main intuition is the following.
Consider an ordinary $k$-connector $\C$ at the end of epoch $t$.
Let $\B_1, \dots, \B_k$ be the sub-rectangles stored at the output cells of 
$\C$.
We would like to define the region  of $\C$ as the intersection of all the sub-rectangles
$\B_1, \dots, \B_k$ since if $\C$ is the main connector of a point $q$, then clearly
$q \in \B_1 \cap \dots \cap \B_k$ and thus the notion of ``region'' will represent the
subset of $\Q$ for which this connector can be useful. 
The same idea can also be defined for the living trees.
Then, our next idea is to define a potential function which is the sum of the  areas of the regions
of all the living trees.
And then, our final maneuver would be to show that this potential function decreases very rapidly every epoch,
and thus by the end of epoch $t$, it is too small to be able to cover all the points $q$ in $\Q$.
However, one technical problem is that to show this decrease in the potential function, we need to 
disregard and discount some small subset of $\Q$; the concept of ``forbidden region'' does this. 
But this also necessitates introducing a 
``dead'' region for each living tree so that we can define the region of a living tree as the intersection
of its output sub-rectangles minus its ``dead'' region. 
And finally, the intuition behind the dead region is the following: consider a point $q$ and a living tree
$T$ in epoch $i$. 
Assume, during epoch $i$ we can already tell that 
no matter how $T$ grows into a main ordinary connector in epoch $t$, $q$ is always
added to the forbidden region. 
In this case, in epoch $i$ we mark $q$ as ``dead'' in advance and thus stop counting it in our potential function.
We now present the actual definitions.

\deff{reg}{
  Consider a living tree $T$ at the end of epoch  $i$ with super-encoding $S$. 
  $T$ will be assigned two geometric regions $\reg(T)$ and $\dead(T)$ and a
    value $\area(T)$. 
    $\area(T)$ is defined as the area of $\reg(T)$. 
    $\reg(T)$ is defined inductively as follows, and it depends on $\dead(T)$ as follows.

    Note that before epoch $i_0$, there are no living trees, as we are only concerned with the
    main ordinary connectors. 
    At epoch $i_0$, the string $S$ exactly encodes the location of the small and large branches that
    belong to this epoch, as well as the locations of their output cells; as $T$ is living, all such output
    cells contain sub-rectangles inserted during epoch $i_0$; $\reg(T)$ is defined as the intersection of all
    such sub-rectangles and here $\dead(T)$ is empty.

    During any later epoch $i$, $i> i_0$, $\reg(T)$ is defined using $\reg(T')$ for a living tree $T'$ in
    epoch $i-1$ from which $T$ grows.
    By definition, $T$ grows out of  $T'$ by adding a number of  small branches that collectively contain 
    $k_i$ small output cells that store sub-rectangles $\B_1, \dots, \B_{k_i}$.
    Then, we may assign a \mdef{dead} region to $T$, denoted by $\dead(T)$.
    $\dead(T)$ always contains the entirety of the current  forbidden region but it might contain more regions. 
    We then define $\reg(T)  = (\reg(T') \setminus \dead(T)) \cap \B_1  \cap \dots \cap B_{k_i}$.
}

We now make the intuition that we had about the dead region into an explicit invariant that we will keep.
\subparagraph{The dead invariant.}
Consider a living tree $T$ in epoch $i$ with super-encoding $S$ and a point $q\in \dead(T)$.
Then, over the choices of the future levels that we may take (in our construction) 
and modifications that the data structure can make to the 
memory graph, it is not possible for $T$ to grow into an ordinary connector connector $\C$ with encoding
$S$ such that $q$ is contained in all the output sub-rectangles of $\C$ but outside the forbidden region. 

We now consider a particular way we can expand the dead region of a living tree.

\deff{dead1}{
    Consider a living tree $T'$ in the beginning of epoch $i$ with super-encoding $S$ and root $\cv$.
    Let $T$ be a living tree at the end of epoch $i$ that has grown out of $T'$. 
    Let $k_i$ be the integer that describes the number of output cells of $T$ on small branches in epoch $i$, as encoded by $S$.

    Observe that all the large branches that belong to epoch $i$ (as encoded by $S$) must be adjacent to $T'$ 
    and furthermore, every output cell of them must be updated with a sub-rectangle inserted
    during epoch $i$ as otherwise no living tree can grow out of $T'$.
    Let $\B_1, \dots, \B_x$ be those sub-rectangles.
    Let $\B'_1, \dots, \B'_y$ be the sub-rectangles inserted during epoch $i$ in the output cells of adjacent small branches of $T'$.
    Observe that we must have $y \ge k_i$ or else no living tree can grow out of $T'$.
    $T$ can grow out of $T'$ by adding all the large branches that contain $\B_1, \dots, \B_x$ and 
    some number of small branches that collectively contain $k_i$ output cells.
    W.l.o.g, assume $\B'_1, \dots, \B'_{k_i}$ are those output cells. 
    Let $\F_i$ be the forbidden region at the end of epoch $i$. 
    We add $\F_i \cup \B'_{k_i+1} \cup \dots \cup \B'_y$ to $\dead(T)$.
}

\begin{ob}\label{ob:dead1}
  The process in \refd{dead1} respects the dead invariant.
\end{ob}
\begin{proof}
  Consider the notation used in the definition. 
  If $q \in \F_i$ then clearly $q$ will be in the forbidden region by the end of epoch $t$ since
  the forbidden region always expands. 
  Let $\B = \B'_{k_i+1}$ and let $\cu$ be the output cell that stores $\B$.
  Thus, consider the case when $q \in \B$; other cases follow identically.
  Observe that by assumptions, $T$ does not include the output vertex $\cu$.
  If $T$ never grows into an ordinary connector with encoding $S$, then we are done.
  Otherwise, according to encoding $S$, at some point $\cu$ should become an output cell.
  This means that at some point in the future, a cell on the connecting path of $\cu$ should be
  updated. 
  However, by the second type of forbidden regions outlined in \refd{f2}, this means that 
  $\B$ will be added to the forbidden region and since $q \in \B$, thus $q$ will be
  added to the forbidden region as well. 
  Thus, our dead invariant is respected.
\end{proof}

\lemm{disjoint}{
  Consider a living tree $T'$ in epoch $i$ and let $T_1, \dots, T_m$ be the living trees in 
  epoch $i+1$ that grow out of $T'$.
  Then, $T_1, \dots, T_m$ have pairwise disjoint regions.
  Also, $\area(T') \le \area(T_1)+ \dots + \area(T_m)$.
}
\begin{proof}
    Consider the notation used in \refd{dead1}.
    Consider the geometric arrangement $\A$ created by the rectangles $\B_1, \dots, \B_x$ and $\B'_1, \dots, \B'_y$.
    Observe for any tree $T_1$ that grows out of $T'$, the region $\reg(T_1)$ is contained in all the
    $\B_1, \dots, \B_x$ and exactly $k_i$ of the sub-rectangles of $\B'_1, \dots, \B'_y$ and furthermore, those $k_i$ 
    sub-rectangles are stored at $k_i$ of the output cells of $T_1$.
    As a result, any other living tree $T_2$ must have a different subset of $k_i$ sub-rectangles at its output cells, meaning,
    its region $\reg(T_2)$ will be disjoint from $\reg(T_1)$, because of the way the dead region is defined in \refd{dead1}.
    This shows that all the regions $\reg(T_1), \dots, \reg(T_2)$ are pairwise disjoint.
    The second part follows since by definition we have $\reg(T_1), \dots, \reg(T_2) \subset \reg(T')$.
\end{proof}

\obs{disjoint}{
  Let $T_1, \dots, T_m$ be living trees with root $\cv$ and super-encoding $S$.
  Then, they have pairwise disjoint regions. 
}
\begin{proof}
  The proof is an easy induction.
  By the definition of living trees and string $S$, there is exactly one living tree at epoch $i_0$. 
  Assume the claim holds for epoch $i$, $i\ge i_0$.
  Then by \refl{disjoint} it also holds for epoch $i+1$.
\end{proof}

The previous observation allows us to simplify and streamline our view of the living trees.
Recall that $\Q = [1,n]\times [1,n]$.
We partition $\Q$ into $n^2$ \idef{pixels} where each pixel is a unit square with integer
coordinates. 
A \idef{pixilated living tree (PLT) $\pT$} is a pair of $(T,p)$ where $T$ is a living tree
and $p$ is a pixel such that $p \subset \reg(T)$.
We define $\reg(\pT) = p$.
By \refo{disjoint}, we can decompose all the living trees with root $\cv$ and with super-encoding
$S$ into a number of PLTs that have disjoint pixels (regions).
By \refo{disjoint} and \refl{disjoint}, if a PLT $\pT'=(T',p')$ grows into a PLT  $\pT=(T,p)$ then $p=p'$ and $T$ is a living tree
that grows out of $T'$ by adding $k_i$ output cells, located on small branches, according to encoding $S$.
In this case, with a slight abuse of the notation we can say $\pT'$ survives until the next epoch, however, it is
more accurate to say that it grows.
However, when $p$ is added to the dead region, forbidden region or if it lies outside the sub-rectangles
inserted at the output cells of $T$, then we have that $p \not \subset \reg(T)$  and thus PLT $\pT'$ does not 
survive the epoch.

\deff{relevant}{
  Consider a fixed string $S$ and a PLT $\pT=(T,p)$ in $\mem^{i-1}$ with super-encoding $S$ and root $\cv$.
  For a cell $\cu$, if the following conditions hold, then we say $\pT$ is \idef{relevant to $\cu$}:
  We must have $\cu \in T_S(\cv,\mem^{i-1})$, $\cu \not \in T$, $\cu$ must be marked as a bud according to 
  string $S$.
}

For technical reasons, for every PLT $\pT=(T,p)$ we need to keep track of a certain value, that we call the \idef{depleted count (DC)}.
Informally, it counts the number
of cells on the adjacent branches of $T$ that can no longer be updated, if $\pT$ were to survive until the end of epoch $t$.
This will be clarified later.
What we need to know at this point is that the DC of a PLT can never decrease, as we progress through the epochs. 
Initially, i.e., at epoch $i_0$, all PLTs have DC 0.

\deff{strings}{
  Consider the beginning of epoch $i$ and let $d$ and $k$ be two integers. 
  Define $\aT_{k,d}$ as the set of all PLTs $\pT$ that have DC $d$ and have a total of
  $k$ output cells in epochs 1 to $i-1$, located on the small
  branches~\footnote{To unclutter the notation, we have not reflected
  the dependency of on the integer $i$.}.
  \ignore{
Let $\S_k$ be the set of all strings $S$ that are the encoding of main ordinary connectors that
  contain $k$ output cells in epochs 1 to $i-1$, located on the small
  branches~\footnote{To unclutter the notation, we have not reflected
  the dependency of $\S_k$ on the integer $i$.}.
  Let $\aT_{k,d}$ be the set of all the PLT that contain $k$ output cells, 
}
  %Let $\S_{k,k_i}$ be the subset of $\S_k$ that contain strings that encode $k_i$ output cells in epoch $i$.
}
\deff{region}{
  Consider a cell $\cu \in \mem^{i-1}$.
Let $\R_u = \left\{ \mbox{$\pT \in \aT_{k,d}$: $\cu$ is relevant to $\pT$} \right\}$.
 For integers $k$, $d$, we define the following:
  \begin{align}
      \regset_{k,d}(\cu) = \bigcup_{\pT \in \R_u} \left\{ \reg(\pT) \right\}\label{eq:regset}
   \end{align}
  \begin{align}
      \asum_{k,d}(\cu) = \sum_{\pT \in \R_u}\area(\reg(\pT))\label{eq:asum}
   \end{align}
  \begin{align}
      \reg_{k,d}(\cu) = \bigcup_{\pT \in \R_u} \reg(\pT) \label{eq:reg}
   \end{align}
   And finally, we set $\area_{k,d}(\cu) = \area(\reg_{k,d}(\cu))$.
}

Note that $\area(\reg(\pT))$  in \refq{asum} is always 1 but we have avoided this substitution to emphasize
that we are computing the total area covered by all the PLTs.

\deff{cellsizes}{
  Consider $\mem^{i-1}$ and two fixed values $k$ and $d$.
  We categorize every cell $\cv \in \mem^{i-1}$
  as one of the following three cases, 
  using two values $V_s = \frac{n2^{\gamma i}}{ U \log^9 n}$, and 
  $V_o = \frac{n2^{\gamma (i+b)} }{U \log^9 n}$.
  \begin{itemize}
    \item If $\area_{k,d}(\cv) < V_s$ then we call $\cv$ a \idef{$(k,d)$-tiny cell}.
    \item If $V_s \le \area_{k,d}(\cv) \le V_o$ then we call $\cv$ a \idef{$(k,d)$-fit cell}.
    \item Finally, if $\area_{k,d}(\cv) > V_o$ then we call $\cv$ a \idef{$(k,d)$-oversized cell}.
  \end{itemize}
}

\deff{func}{
  Consider a cell $\cv$.
  We define a function \mdef{$f_{(k,d),\cv}(\cdot)$}
  from $\Q$ to $\N$ where for a point $q \in \Q$, $f_{(k,d),\cv}(q)$ is the number
  of regions in $\regset_{k,d}(\cv)$ that contain the point $q$.
  The \idef{dense} subset of $\reg_{k,d}(\cv)$, denoted by $\dense_{k,d}(\cv)$,  is the subset of $\Q$ of area
  at most $V_o$ that contains the largest values of the $f_{(k,d),\cv}$ function, meaning,
  for every point $q \in \dense(\cv)$ and every point $q' \not \in \dense_{k,d}(\cv)$ we have
  $f_{(k,d),\cv}(q) \ge f_{(k,d),\cv}(q')$.

  If $\cv$ is not $(k,d)$-oversized, then $\dense_{k,d}(\cv) = \reg_{k,d}(\cv)$ but if 
  $\cv$ is $(k,d)$-oversized, then $\dense_{k,d}(\cv)$ has area $V_o$.
}
Finally, the following operations are performed at the beginning of each epoch. 
\deff{cfing}{
  At the beginning of epoch $i$, and for every $\cv \in \mem^{i-1}$ do the following. 
    \begin{itemize}
      \item (default marking) If  for some $k$ and $d$, $k,d\le \log^2 n$,
        $\cv$ is a $(k,d)$-tiny cell, then we conditionally forbid
        $\reg_{k,d}(\cv)$ at $\cv$. 
        If $\cv$ is $(k,d)$-fit, we conditionally forbid $\reg_{k,d}(\cv)$ at $\cv$
        at epoch $i+b$.
        If $\cv$ is $(k,d)$-oversized, we conditionally forbid $\dense{k,d}(\cv)$ at $\cv$
        at epoch $i+b$.

        \item (fit marking) If for some $k$ and $d$, $k,d\le \log^2 n$, $\cv$ is a $(k,d)$-fit cell, we do the following. 
          At the beginning of epoch $i$, we ``assign'' $\reg_{k,d}(\cv)$ to $\cv$ to be later conditionally forbidden. 
          Thus, at the beginning of epoch $i$, each cell will have an ``assigned'' region. 
          Then, in epochs $i$ until $i+2b$, every time the data structure updates a cell $\cv_2$, 
            all the regions that are assigned to cells in the $\mu$-out-neighborhood of $\cv_2$ are added to the
            region assigned to $\cv_2$.
            At epoch $i+2b$, all the regions assigned to the cells are conditionally forbidden. 

        \item (dense marking) If for some $k$ and $d$, $k,d\le \log^2 n$, $\cv$ is an $(k,d)$-oversized cell, then we repeat the above process (the process for the fit cells)
          but only for $\dense_{k,d}(\cv)$.
    \end{itemize}
}

We now bound the total area we conditionally forbid at each epoch using the above processes.
First consider the default marking process. 
If $\cv$ is $(k,d)$-tiny, then it CF an area of at most $O(\log^2 nV_s)$ over all choices of $k$ and $d$.
Otherwise, an area of at most $O(\log^2 nV_s)$ is CF, over all choices of $k$ and $d$ but at epoch $i+b$
but this is allowed by Lemma~\ref{lem:cf}.
The fit marking is a bit more complicated.
Observe that at the beginning of epoch $i$, each cell is assigned an area of at most $O(\log^2 nV_o)$ 
(over all choices of $k$ and $d$).
Whenever the algorithm updates a node $\cu$, then the regions assigned to all the cells in the $\mu$-out-neighborhood
of $\cu$ are added to the assigned region of $\cu$; this can increase the assigned area of $\cu$
by at most a factor $2^{\mu}$ (the size of $\mu$-out-neighborhood of $\cu$).
This process continues for $2b$ epochs, meaning, by epoch $i+2b$, the assigned area of each node can be
as large as $O(\log^2 n V_o 2^{\mu 2b})$.
However, observe that these cells are conditionally forbidden at epoch $i'= i+2b$.
But in epoch $i'$, we are allowed to have 
an area of $\Theta\left(\frac{2^{\gamma i'} n}{U\log^4 n}
\right)$ conditionally forbidden at each cell, by Lemma~\ref{lem:cf}.
Observe that
\begin{align*}
  \log^2n V_o 2^{\mu 2b} &=  \log^2n \frac{n2^{\gamma (i+b)} }{U \log^9 n}2^{\mu 2b} \\
  &<   \frac{n2^{\gamma (i+b)} }{U \log^7 n}2^{ \gamma b}  = o\left(\frac{2^{\gamma i'} n}{U\log^4 n}
\right)
\end{align*}
and thus this is well within the condition specified in Lemma~\ref{lem:cf}.
The case of dense marking is exactly identical to this case.

%%%%%%%%%%%%%%%%%%%%%%%%%%%%%%%%%%%%%%%%%
%                                       %
%    POTENTIAL FUNCTION ARGUMENT        %
%                                       %
%%%%%%%%%%%%%%%%%%%%%%%%%%%%%%%%%%%%%%%%%

\subsection{A Potential Function Argument.}
Our final analysis is in the form of a potential function argument.
We in fact define a series of potential functions.

\deff{pot}{
  For integers $k$ and $d$, $1\le k \le t$, and in $\mem^{i-1}$ (i.e., the beginning of epoch $i$)
  we define a potential function $\Phi(i-1,k,d)$ as follows:
  \begin{align}
    \Phi(i-1,k,d) = \sum_{\pT \in \aT_{k,d}} \area(\reg(\pT)).\label{eq:pot}
  \end{align}
}

\obs{potob}{
  We have 
  \[
    \sum_{\cu \in \mem^{i-1}} \asum_{k,d}(\cu) = O(t \Phi(i-1,k,d)).
  \]
}
\begin{proof}
  Consider a PLT $\pT \in \aT_{k,d}$ (i.e., with $k$ output cells and DC $d$).
  The area of $\reg(\pT)$ is added to $\asum_{k,d}(\cu)$, for $O(t)$ cells $\cu$
  since $\pT$ has $O(t)$ adjacent buds.
\end{proof}

However, before getting to the potential function analysis, we need to investigate how a living tree
in epoch $i$ can grow into a living tree in the future epochs. 
Let $\pT=(T,p) \in \aT_{k,d}$ be a PLT tree with super-encoding $S$ at the beginning of epoch $i$ with DC $d$.
By definition, $T$ has $k$ output cells, located on small branches. 
We would like to calculate the probability of $\pT$ surviving in the next epoch, assuming it needs
to add $k_i$ output cells that are located
on the small branches, where $k_i$ is the integer encoded by string $S$.
Note that we completely ignore the long branches; the string $S$ exactly encodes their location and 
we assume the data structure always places a suitable sub-rectangle there. 

Let $\cu_1, \dots, \cu_x$ be the adjacent buds of $T$ that are located on the small branches.
We have essentially four cases and we analyze each in a different subsection.

\subsubsection{An adjacent small bud of $T$ is $(k,d)$-tiny.}\label{sec:tiny}
\begin{lemma}\label{lem:tiny}
  If any of the cells $\cu_1, \dots, \cu_x$ is $(k,d)$-tiny, then $\pT$ cannot survive
  until epoch $t$.
\end{lemma}
\begin{proof}
  W.l.o.g., assume $\cu_1$ is $(k,d)$-tiny. 
  $\cu_1$ is located on an adjacent small branch which means
  its connecting path has length at most $\mu$. 
  In addition, recall that the region of every tiny cell has been conditionally
  forbidden at the beginning of the epoch (by the default marking process)
  Since, $\reg(\pT)$ is contained in $\reg_{k,d}(\cu_1)$, it thus follows that $\reg(\pT)$ has been conditionally forbidden.
  As a consequence, if the data structure updates any of the cells on the 
  connecting path of $\cu_1$, $\reg(\pT)$ is added to the forbidden region.
  On the other hand, according to string $S$, $\cu_1$ 
  is marked as a bud and thus eventually it should be updated to lead to an output cell, a contradiction. 
\end{proof}

Thus, in this case, we immediately and  in advance add $\reg(\pT)$ to the dead region, consistent with the dead invariant.
As a consequence, we can assume that none of the cells $\cu_1, \dots, \cu_{x}$ is 
$(k,d)$-tiny, meaning, they are either $(k,d)$-fit or $(k,d)$-oversized.

\subsubsection{Too many fit cells and dense areas.}\label{sec:toomany}
\begin{lemma}\label{lem:toomany}
  Assume there exists at least $y=t^{0.1}+2b^2+b$ cells $\cu_j$, among the cells  $\cu_1, \dots, \cu_x$,
  such that either  (i) $\cu_j$ is $(k,d)$-fit or
  (ii) $\cu_j$ is $(k,d)$-oversized and $\reg(\pT) \in \dense_{k,d}(\cu_j)$.
  Assume $\pT$ survives the epoch $i+2b$.
  Then, there are at least $t^{0.1}$ cells, $\cu_j$, among the cells $\cu_1, \dots, \cu_x$
  with following property:
  $\cu_j$ has a cell $\cu'_j$ on its connecting path such that if any of the cells on the
  connecting path of $\cu'_j$ is updated, then $\reg_(\pT)$ is added to the forbidden region. 
\end{lemma}
\begin{proof}
  W.l.o.g, let $\cu_1, \dots, \cu_y$ be the cells that satisfy condition (i) or (ii) outlined
  in the lemma. 

  Consider the epochs $i$ to $i+2b$.
  We claim, during each of these epochs, $\pT$ can grow by at most $b$ output cells; assume
  during one epoch $\pT$ grow by at least $b+1$ output cells, that store $b+1$ sub-rectangles.
  However, in each epoch, we take exactly $b$ levels, one in each odd band.
  This implies that at least two of the sub-rectangles must be in the same level and thus they have
  empty intersection. 
  As a result, it follows that $\pT$ grows into a PLT with an empty region, a contradiction. 
  Thus, during epochs $i$ to $i+2b$, $\pT$ can grow by at most $2b^2+b$ output cells, leaving
  at least $t^{0.1}$ cells among $\cu_1, \dots, \cu_y$ still in the adjacent branches. 

  Consider one of those cells $\cu_j$ and  consider the fit marking and the dense marking process in \refd{cfing}.
  Consider the connecting path $\pi_j$ of $\cu_j$ in epoch $i$.
  Let $\cu'_j$ be the cell farthest from $\cu_j$ in $\pi_j$ that is updated during epochs $i$ to $i+2b$.
  Recall that by the default marking process, if $\cu$ is $(k,d)$-fit (resp.
  $(k,d)$-oversized) we CF $\reg_{k,d}(\cu)$ (resp. $\dense_{k,d}(\cu)$) at
  $\cu$ and also recall that by our assumptions
  $\reg(\pT) \subset \reg_{k,d}(\cu)$ (resp. $\reg(\pT) \subset \dense_{k,d}(\cu)$).
  As a result, 
  if $\cu'_j$ does not exist, then by the default marking process, $\reg(\pT)$ is CF at $\cu$ and thus
  the data structure cannot update $\cu$ or any of the cells in its connecting path without
  adding $\reg(\pT)$ to the forbidden region, meaning, $\pT$ cannot survive until epoch $t$.
  Otherwise, by the fit marking and the dense marking processes, $\reg(\pT)$ is conditionally forbidden at
  epoch $i+2b$ at $\cu'_j$. 
  The cells $\cu'_j$ are the cells claimed in the lemma.
\end{proof}

In this case, we increase the DC of $\pT$ and this is the only way the DC of a PLT can increase:
at epoch $i+2b$, we can consider the cells $\cu'_j$ as ``depleted'' and  to reflect that,
the DC of $\pT$ is increased by $t^{0.1}$.
One technical issue that we need to outline here is the DC does not count the number of distinct depleted cells as 
we could have $\cu'_{j} = \cu'_{\ell}$ for to different cells $\cu_j$ and $\cu_\ell$.
Instead, DC counts the number of depleted cells with multiplicity. 
To make the later analysis easier, we can bound this as follows.
If a PLT $\pT$ satisfies the conditions in \refl{toomany}, we increase its DC by $t^{0.1}$
but assume it can grow without any restrictions until epoch $i+2b$.
This idea will come handy later when we are trying to analyze our potential function.

\subsubsection{Growth Using Fit or Dense Areas}\label{sec:fit}
\begin{lemma}\label{lem:fit}
  Consider a PLT $\pT \in \aT_{k,d}$ with super-encoding $S$.
  Consider a point $q \in \reg(T)$ and assume it has not been marked dead in this epoch.
  Consider the cells $\cu_j$
  such that either  (i) $\cu_j$ is $k$-fit or
  (ii) $\cu_j$ is $(k,d)$-oversized and $q \in \dense_{k,d}(\cu)$.
  W.l.o.g, let $\cu_1, \dots, \cu_y$ be all such cells.

  When $y<t^{0.1}+2b^2+b$,
  the probability that $\pT$ survives by the end of epoch $i$ and grows by $k_i$ additional output cells,
  using only cells $\cu_1, \dots, \cu_y$  is at most $t^{-k_i/2}$;
  $k_i$ is the number of output cells that must belong to epoch $i$, on the small branches, as encoded by $S$. 
\end{lemma}
\begin{proof}
  Let $p=\reg(\pT)$ be the pixel associated with $\pT$.
  Let us review our progress: by Lemma~\ref{lem:tiny}, none of the cells $\cu_1, \dots, \cu_x$ can be
  $(k,d)$-tiny.

  Consider a cell $\cu$, among the cells $\cu_1, \dots, \cu_y$.
  Since we are working with encoding $S$, for $\pT$ to grow using $\cu$,
  $\cu$ should become a bud, meaning, 
  the data structure must update
  $\cu$ but none of the cells on the connecting path of $\cu$.
  Furthermore, after updating $\cu$, the data structure has to store a sub-rectangle $\B$ at the position of
  the output cell of $\cu$, with $p\in \B$ (if $p \not \in \B$ then $\pT$ does not survive,
  by the definition of $\reg(\pT)$). 
  This in turn implies that $\B$ exists as a tag (see Definition~\ref{def:f1} on page \pageref{def:f1}) in the arrangement of tags, $\A_{\cu}$, defined 
  during the first type of forbidden regions. 

  Next, recall that when we defined the first type of forbidden regions, we
  added any point that is contained in more
  than $10b\log t$ tags to the forbidden region.
  As a result, the sub-rectangle $\B$ can only be one of the $10b \log t$ possible sub-rectangles 
  in $\A_{\cu}$ that contain the pixel $p$.
  Let $\cS_\cu$ be the set of sub-rectangles in the arrangement $\A_{\cu}$ that contain the pixel $p$. 

  It now remains to make one crucial but almost trivial observation:
  a necessary condition for a sub-rectangle $\B_j$  to be stored anywhere
  by the data structure is that it must be inserted in the insertion sequence!
  And a necessary condition for the latter is that
  a particular level must be taken, among the untaken levels in our construction.
  Furthermore, as any set of $k_i$ distinct sub-rectangles with non-empty intersection correspond to
  sub-rectangles inserted on $k_i$ distinct odd bands, it follows that the probability
  that we insert any fixed set of $k_i$ sub-rectangles in our insertion sequence is 
  at most $t^{-k_i}$, by \refo{choices}.

  The number of different ways we can select $k_i$ sets among the sets
  $\cS_{\cu_1}, \dots, \cS_{\cu_y}$ and then select $k_i$ distinct sub-rectangles, one from each selected set, 
  is at most ${y \choose k_i} (10b\log t)^{k_i}$ and thus the probability that $p$ is contained
  in the intersection of $k_i$ sub-rectangles is bounded by 
  \begin{align*}
    {y \choose k_i} (10b \log t)^{k_i} \cdot t^{-k_i} \le &(\Theta({\frac{y}{k_i}}))^{k_i} (10b \log t)^{k_i} \cdot t^{-k_i} \le \\
    &\left( \frac{\Theta(b^3t^{0.1} \log t)}{t} \right)^{k_i}  \le t^{-k_i/2}
  \end{align*}
  according to our choice of parameters. 
\end{proof}

\subsubsection{Using Oversized Cells.}\label{sec:ov}
Consider the memory graph $\mem^{i-1}$ and consider a $(k,d)$-oversized cell $\cu$. 
We bound how much the ``not dense part of'' $\reg_{k,d}(\cu)$ can help the growth of the PLTs. 

\lemm{ov}{
  Consider the beginning of epoch $i$ and the graph $\mem^{i-1}$.
  Let $\cu \in \mem^{i-1}$ be a $(k,d)$-oversized cell.
  Let $\aT_\cu$ contain all the PLTs $\pT$ that satisfy the following: $\pT$ is a PLT
  at the end of epoch $i$, $\pT$  has grown out of a PLT $\pT'$ (at the beginning of epoch $i$)
  by adding a branch with $\cu$ as a bud, 
  and $\reg(\pT) \subset \reg_{k,d}(\cu) \setminus \dense_{k,d}(\cu)$ (here, $\reg_{k,d}(\cu)$ and $\dense_{k,d}(\cu)$
  are considered at the beginning of epoch $i$).
  Then, $\sum_{\cu \in \mem^{i-1}} \sum_{\pT \in \aT_\cu} \area(\reg(\pT)) \le \frac{\Phi(i-1,k,d)}{2^{\gamma(b-1)} }$.
}

\begin{proof}
  Let $\B_\cu$ be the union of all the sub-rectangles in the $\mu$-out-neighborhood of $\cu$ inserted during epoch $i$.
  We begin by bounding the following:
\begin{align}
  \sum_{\pT \in \aT_\cu} \area(\reg(\pT)). \label{eq:ovgrowth}
\end{align}
Consider a PLT $\pT$ (at the end of epoch $i$) with super-encoding $S$ that has grown out of a PLT $\pT'$
(at the beginning of epoch $i$), by adding $\cu$ as a bud. 
Observe that in $\pT$, bud $\cu$ uniquely identifies a memory cell in the 
$\mu$-out-neighborhood of $\cu$ that is the output cell which stores a sub-rectangle $\B_T$ inserted during epoch $i$.
By definition of $\reg(\pT)$ and $\reg(\pT')$, we have $\reg(\pT) = \reg(\pT') \subset \reg_{k,d}(\cu)\setminus \dense_{k,d}(\cu)$. 
But we must also have $\reg(\pT) \subset \B_T$.
As a result, $\reg(\pT) \subset ( \reg_{k,d}(\cu)  \setminus \dense_{k,d}(\cu)) \cap \B_T = 
 \reg_{k,d}(\cu) \cap (\B_T  \setminus \dense_{k,d}(\cu))$. 

Recall that by definition of function $f_{(k,d),\cu}(\cdot)$,  
for a point $q \in \B_\cu \setminus\dense_{k,d}(\cu)$, $f_{(k,d),\cu}(q)$ counts how many living tree $T'$ have $ q\in \reg(T')$.
As a result, we can rewrite Eq.~\ref{eq:ovgrowth} as follows: 
  \begin{align}
    \mbox{ (\ref{eq:ovgrowth}) }  \le \int_{\B_\cu\setminus \dense_{k,d}(\cu)} f_{(k,d),\cu}(q)d q.\label{eq:_barea}
  \end{align}
  %First, remember that the area of $\B$ is $\Theta\left(2^{\gamma i} n   \right)$.
  As $\cu$ is $(k,d)$-oversized, we have, $\asum_{k,d}(\cu) \ge V_o$.
  By definition of $f_{(k,d),\cu}$, 
  \begin{align}
    \int_\Q f_{(k,d),\cu}(q)d q = \asum_{k,d}(\cu) \label{eq:fq}
  \end{align}
  We rewrite (\ref{eq:fq}) as
  \begin{align*}
      &\int_{\Q\setminus \dense_{k,d}(\cu)} f_{(k,d),\cu}(q)d q +  \int_{\dense_{k,d}(\cu)} f_{(k,d),\cu}(q)d q\\
     & = \asum_{k,d}(\cu)
  \end{align*}
  Next, observe that for every point $q' \in \dense_{k,d}(\cu)$ and every point $q \in \Q\setminus\dense_{k,d}(\cu)$
  we have $f_{(k,d),\cu}(q') \ge f_{(k,d),\cu}(q)$, by the definition of $\dense_{k,d}(\cu)$.
  As a result, 
  \begin{align*}
    &\frac{\int_{\dense_{k,d}(\cu)} f_{(k,d),\cu}(q)d q}{\int_{\B_\cu\setminus \dense_{k,d}(\cu)} f_{(k,d),\cu}(q)d q} \ge \frac{\area(\dense_{k,d}(\cu))}{\area(\B_\cu\setminus \dense_{k,d}(\cu))} \\
    &\ge \frac{V_o}{\area(\B_\cu)}.
\end{align*}

% <-------------------------------------- IN FULL PAPER ---------------------------------------------->
\infull{
      Thus, 
      \begin{align}
        \mbox{ (\ref{eq:_barea})} \le \frac{\area(\B_\cu)\int_{\small \dense_{k,d}(\cu)} f_{(k,d),\cu}(q)d q}{V_o} \le \frac{\area(\B_\cu)\int_{\Q} f_{(k,d),\cu}(q)d q}{V_o} = \frac{ \area(\B_\cu) \asum_{k,d}(\cu)}{V_o}.\label{eq:barea}
      \end{align}

    Thus, we have
    \begin{align*}
      \mbox{ (\ref{eq:ovgrowth}) }   \le  \frac{ \area(\B_\cu) \asum_{k,d}(\cu)}{V_o}.
    \end{align*}
    By Lemma~\ref{lem:inssize}, each sub-rectangle inserted during epoch $i$ has area  $\Theta\left(2^{\gamma i} n   \right)$
    and $\B_\cu$ is the union of at most $2^{\mu}$ of them. 
    Thus, $\area(\B_\cu) = O(2^{\mu} 2^{\gamma i} n)$.
    By plugging in the value for $V_o$ we get,
    \begin{align*}
      \mbox{ (\ref{eq:ovgrowth}) }   &\le 
      O\left( \frac{ 2^{\mu} 2^{\gamma i} n \asum_{k,d}(\cu)}{\frac{n2^{\gamma (i+b)} }{U \log^9 n}} \right) = 
      O\left( \frac{ 2^{\mu}{U \log^9 n} \asum_{k,d}(\cu)}{2^{\gamma b} } \right) \\
      & \le O\left( \frac{ 2^{\mu}{U \log^{11} n} \asum_{k,d}(\cu)}{t^2 2^{\gamma b} } \right) <  \frac{\asum_{k,d}(\cu)}{t^2 2^{\gamma(b-1)} }.
    \end{align*}
    The lemma then follows by summing the above over all cells $\cu \in\mem^{i-1}$ and using \refo{potob}.
}{
% <-------------------------------------- IN DRAFT ---------------------------------------------->
    The lemma then follows by 
    combining this inequality with Eq.~\ref{eq:fq}, and plugging values for the area of $\B_u$ and
    $V_o$.
    For more details, see the full version. 
}
\end{proof}

\subsection{The Potential Function Analysis.}
Let $\bL_{i-1}$ be the random variable that represents the set of choices (of which level to take)
during epochs 1 to $i-1$ and let $L_{i-1}$ be one particular set of choices.
Recall that the potential function was defined as
\[
    \Phi(i,k,d) = \sum_{\pT \in \aT_{k,d}} \area(\reg(\pT)).
\]
Note that the potential function is in fact a random variable that depends on $\bL_{i-1}$.
We use the bold math font to highlight this, e.g., when the potential is a random variable
we represent it with $\bphi(i,k,d)$.
Thus, $\bphi(i-1,k,d)|\bL_{i-1} = L_{i-1}$ is simply a fixed value and not a random variable. 
Let $X_{i',k',d'} = \bphi(i',k',d')|\bL_{i-1} = L_{i-1}$ for $i' < i$. 
We estimate $\E[\bphi(i,k,d)| \bL_{i-1} = L_{i-1}]$ based on these values. 

Consider a PLT $\pT$ at the end of epoch $i$ that is counted in $\bphi(i,k,d)$.
We have the following cases:
\begin{itemize}
  \item (case 1: via \refl{toomany}) If there is an epoch $i'$, $i-2b \le i' < i$,
    and a PLT $\pT'$ in epoch $i'$ that satisfied the conditions in \refl{toomany} such that
    $\pT$ is grown out of $\pT'$, then we count $\pT$ in this case.
    However, in this case,  $\pT'$ must have had DC $d-t^{0.1}$
    but $\pT'$ could have as few as $k-2b^2$ output cells. 
    Let $\bphi_1$ be the part of  $\bphi(i,k,d)$ that falls in this case.
    Thus,
\begin{align}
  \E[\bphi_1| \bL_{i-1} = L_{i-1}] \le \sum_{1 \le j_1 \le 2b; 1 \le j_2 \le 2b^2}X_{i-j_1,k-j_2,d-t^{0.1}}.\label{eq:pot1}
\end{align}

  \item (case 2: via \refl{fit}) Otherwise, if $\pT$ is grown out of a PLT $\pT'$ in epoch $i-1$ that satisfied
    the conditions of \refl{fit}, then we count $\pT$ in this case.
    Let $\bphi_2$ be the part of the potential that falls in this case.
    By \refl{fit} we have,
  \begin{align}
    \E[\bphi_2 | \bL_{i-1} = L_{i-1}] \le \sum_{j=0}^{b} \frac{X_{i-1,k-j,d}}{t^{j/2}}.\label{eq:pot2}
  \end{align}
  Note that this case also counts the case when $\pT= \pT'$, i.e., $\pT'$ does not grow and does not gain
  additional output cells. 
  \item (case 3: via \refl{ov}) The final case is when, $\pT$ is grown by the conditions specified in \refl{ov}.
    Let $\bphi_3$ be the part of the potential that falls in this case.
    By \refl{ov} we have,
 \begin{align}
   \E[\bphi_3 | \bL_{i-1} = L_{i-1}] \le \sum_{j=1}^{b}  \frac{X_{i-1,k-j,d}}{2^{\gamma(b-1)} }.\label{eq:pot3}
  \end{align}
\end{itemize}

Observe that the right hand side of \refq{pot3} is asymptotically bounded by the right hand side of \refq{pot2}.
Thus, by adding up \refq{pot1}, \refq{pot2}, and \refq{pot3} and taking the expectation over the choices of $\bL_{i-1}$,
we get the following recursion
% <-------------------------------------- IN FULL ---------------------------------------------->
\infull{
    \begin{align}
      \E[\bphi(i,k,d)] \le \sum_{j=0}^{b} \frac{\E[\bphi(i-1,k-j,d)]}{(\varepsilon_0 t)^{j/2}} +  \sum_{1 \le j_1 \le 2b; 1 \le j_2 \le 2b^2}\E[\bphi(i-j_1,k-j_2,d-t^{0.1})]\label{eq:rec}
    \end{align}
}{
% <-------------------------------------- IN DRAFT ---------------------------------------------->
    \begin{align*}
      \E[\bphi(i,k,d)] \le & \sum_{j=0}^{b} \frac{\E[\bphi(i-1,k-j,d)]}{(\varepsilon_0 t)^{j/2}} +  \\
      &\sum_{1 \le j_1 \le 2b; 1 \le j_2 \le 2b^2}\E[\bphi(i-j_1,k-j_2,d-t^{0.1})]
    \end{align*}
}
where $\varepsilon_0$ is some constant (it comes from the fact that we have absorbed \refq{pot3} into \refq{pot2}).

Now it is time to revisit the concept of main ordinary connectors and see what are the initial values of these recursive functions. 
Recall that any main ordinary connector has $k_0$ output cells from epoch $i_0$ and we have
$i_0 \le \frac{6t\alpha}{\beta}$.
Now consider epoch $i_0$.
We have $\Phi(i_0,k_0,d) = 0$, as long as $d\not = 0$.
So we upper bound $\Phi(i_0, k_0,0)$. 
To that end, we simply count the maximum number of living trees $T$ we can have in epoch $i_0$.
We have $\Theta(nU)$ choices for the root of $T$. 
The number of strings $S$ is at most $2^{O(t \beta)}$ by Lemma~\ref{lem:enc}.
Finally, $\reg(T)$ of each living tree is at most $\Theta(n2^{\gamma i_0})$, since
by Lemma~\ref{lem:inssize}, the area of any sub-rectangle inserted during epoch $i_0$ is $\Theta(n2^{\gamma i_0})$.
Let $M= \Theta(n2^{\gamma i_0}  2^{O(t \beta)} nU)$.

Thus, we can obtain the following initial values for the potential functions:
\begin{eqnarray}
    \Phi(i,k,d) = 
    \begin{cases}
        0 & d\not = 0 \\
        0 & i<i_0 \\
        0 & i=i_0, k\not = k_0 \\
        M & i=i_0, k=k_0, d=0
    \end{cases}
    \label{eq:init}
\end{eqnarray}

% <-------------------------------------- IN FULL ------------------------------------------->
\infull{
    Solving \refq{rec} is a bit difficult. However, we can use induction.
    First, we consider the case when $d=0$ and observe that the 2nd summation disappears (since DC cannot be negative) and we are left with the following recursion:
    \begin{align}
      \E[\bphi(i,k,0)] \le \sum_{j=0}^{b} \frac{\E[\bphi(i-1,k-j,0)]}{(\varepsilon_0 t)^{j/2}} \le  \sum_{j=0}^{k-1} \frac{\E[\bphi(i-1,k-j,0)]}{(\varepsilon_0 t)^{j/2}}.\label{eq:rec2}
    \end{align}
    We guess that
    \begin{align}
      \E[\Phi(i_0 + i,k_0 + k,0)] \le \frac{ M {i+k-2 \choose i-1}}{ (\varepsilon_0 t)^{k/2}}\label{eq:sol}.
    \end{align}
    We can simply plug this guess and verify:
    \begin{align}
      \E[\Phi(i_0 + i,k_0 + k,0)] \le  \sum_{j=0}^{k-1}  \frac{ M {i+k-j-3 \choose i-2}}{ (\varepsilon_0 t)^{(k-j)/2} {(\varepsilon_0 t)^{j/2}}} = \sum_{j=0}^{k-1}  \frac{ M {i+k-j-3 \choose i-2}}{ (\varepsilon_0 t)^{k/2} }=\frac{ M {i+k-2 \choose i-1}}{ (\varepsilon_0 t)^{k/2}}
    \end{align}
    where in the last step, we are using that 
    \[
      {k \choose k }+ {k+1\choose k } + \dots + {n\choose k} = {n+1 \choose k+1}
    \]
    is a known binomial identity~\footnote{For a quick proof, observe that the left hand size counts how many ways one can select $k+1$ numbers among numbers $1$ to $n+1$ by first selecting the maximum value, $i$,
    then selecting the $k$ remaining elements from $1$ to $i-1$.}.
    Having the value of $\E[\bphi(i_0 + i, k_0 + k, 0)]$, we can now work out the value of
    $\E[\bphi(i_0 + i, k_0 + k, t^{0.1})]$, as we can now replace a value for the second summation in \refq{rec}.
    This forms the basis of our induction argument.
    We guess that
    \begin{align}
      \E[\Phi(i_0 + i,k_0 + k,dt^{0.1})] \le \frac{2^{i+k+d} M {i+k-2 \choose i-1}}{ (\varepsilon_0 t)^{(k-2b^2d)/2}}\label{eq:sol2}
    \end{align}
    and thus we would like to bound $\E[\Phi(i_0 + i,k_0 + k,(d+1)t^{0.1})]$.
    \begin{align}
      \E[\bphi(i,k,(d+1)t^{0.1})] &\le \sum_{j=0}^{b} \frac{\E[\bphi(i-1,k-j,(d+1)t^{0.1}d)]}{(\varepsilon_0 t)^{j/2}}   \label{eq:t1} \\
      &+  \sum_{1 \le j_1 \le 2b; 1 \le j_2 \le 2b^2}\E[\bphi(i-j_1,k-j_2,dt^{0.1})] \label{eq:t2}
    \end{align}

    We bound each of the sums separately and for clarity:
    \begin{align*}
      \mbox{ (\ref{eq:t1}) } &\le  \sum_{j=0}^{k-1} \frac{\E[\bphi(i-1,k-j,(d+1)t^{0.1}d)]}{(\varepsilon_0 t)^{j/2}}  \le    \sum_{j=0}^{k-1}2^{i-1+k-j+d+1} \frac{M {i-1+k-j-1 \choose i-2}}{(\varepsilon_0 t)^{(k-j-2b^2(d+1))/2}(\varepsilon_0 t)^{j/2}}  \\
      & \le   2^{i+k+d}M  \sum_{j=0}^{k-1} \frac{ {i-1+k-j-1 \choose i-2}}{(\varepsilon_0 t)^{(k-2b^2(d+1))/2}}  \le   2^{i+k+d} M  \frac{ {i-1+k \choose i-1}}{(\varepsilon_0 t)^{(k-2b^2(d+1))/2}}.  \\
    \end{align*}

    \begin{align*}
      \mbox{ (\ref{eq:t2}) } &\le  \sum_{1 \le j_1 \le 2b} \sum_{1 \le j_2 \le 2b^2}  \frac{2^{i-j_1+k-j_2+d} M {i-j_1 -1+k-j_2-1 \choose i-j_1-1}}{ (\varepsilon_0 t)^{(k-j_2-2b^2d)/2}} \\
      &\le \frac{2^{i+k+d}M}{ (\varepsilon_0 t)^{(k-2b^2-2b^2d)/2}}\sum_{1 \le j_1 \le 2b} \sum_{1 \le j_2 \le 2b^2}  {i-j_1 -1+k-j_2-1 \choose i-j_1-1} \\
      &\le \frac{2^{i+k+d}M}{ (\varepsilon_0 t)^{(k-2b^2-2b^2d)/2}}\sum_{1 \le j_1 \le 2b}  {i-j_1 -1+k-1 \choose i-j_1} \\
      &= \frac{2^{i+k+d}M}{ (\varepsilon_0 t)^{(k-2b^2(d+1))/2}}\sum_{1 \le j_1 \le 2b}  {i-j_1 -1+k-1 \choose k-2} \le \frac{2^{i+k+d}M}{ (\varepsilon_0 t)^{(k-2b^2(d+1))/2}} {i-1+k-1 \choose k-1}.  \\
    \end{align*}

    Thus, 
    \begin{align}
      \mbox{ (\ref{eq:t1}) } + \mbox{ (\ref{eq:t2}) } \le  \frac{2^{i+k+d+1}M}{ (\varepsilon_0 t)^{(k-2b^2(d+1))/2}} {i-1+k-1 \choose k-1}\label{eq:finalsol} 
    \end{align}
    as claimed.
}{
% <-------------------------------------- IN DRAFT ------------------------------------------->
    In the full version of the paper~\cite{Afshani.dyfc.arxiv}, we show that we
    can solve this recursion using induction and 
    guessing the answer.  Our analysis shows that the 
    \begin{align}
     \E[\bphi(i,k,d)]  \le  \frac{2^{i+k+d+1}M}{ (\varepsilon_0 t)^{(k-2b^2(d+1))/2}} {i-1+k-1 \choose k-1}\label{eq:finalsol}.
    \end{align}
}

We are now almost done.
Let $\Q'$ be the subset of $\Q$ that lies outside the forbidden and strange regions. 
Recall that the total area of the forbidden and the strange regions is $o(n^2)$ and thus the
area of $\Q'$ is $\Omega(n^2)$.
Because of the dead invariant, any point $q \in \Q'$, must be contained in the $\reg(\C)$ where $\C$ is a main connector of $q$ with some
values of $i_0$ and $k_0$.
Consequently, there are values of $i_0$ and $k_0$ for which the total area of $\reg(\C)$ over all
the main connectors that have exactly $k_0$ outputs from epoch $i_0$ is at least $\Omega(n^2/\log^2n)$.
In addition, any such connector must have at least $2t$ output cells and at least $t$ output cells that lie on small branches. 
Also, for each connector the DC can be at most $2t\mu$ and thus the parameter $d$ in \refq{finalsol} can be at most
$\frac{2t\mu}{t^{0.1}} = 2t^{0.9}\mu < t^{0.91}$.
Thus, by \refq{finalsol}, the expected area of such connectors is bounded by $\Phi(t,t,dt^{0.1}) = \Phi(i_0 + (t-i_0), k_0+ (t-k_0),d t^{0.1})$.
As $k_0 \le b = \Theta(\log t)$, and $i_0 \le \frac{6t\alpha}{\beta}$,
 it follows that
\infull{
 \begin{align*}
   \Phi(t,t, dt^{0.1}) & \le \Phi(i_0 + t, t/2, t^{0.91}t^{0.1}) \le \frac{ 2^{O(t)} \Theta(n2^{\gamma i_0}  2^{O(t
   \beta)} nU)  {3t/2 \choose t/2}}{ (\varepsilon t)^{(t-2b^2t^{0.9})/4}} = \frac{ n^2
    2^{\gamma i_0}  2^{O(t \beta)}}{ (\varepsilon t)^{t/8}}  \\
    &\le  
    \frac{ n^2
      2^{\gamma \frac{6t\alpha}{\beta}}  2^{O(t \beta)}}{2^{\Omega(t \log\log n)}}   
       \le \frac{ n^2 2^{6 \varepsilon \gamma t}  2^{O(t \beta)}}{2^{\Omega(t \log\log n)}}  \le  \frac{ n^2 }{2^{\Omega(t \log\log n)}}  
  \end{align*}
}{
 \begin{align*}
   &\Phi(t,t, dt^{0.1})  \le \Phi(i_0 + t, t/2, t^{0.91}t^{0.1})  \\
   & \le \frac{ 2^{O(t)} \Theta(n2^{\gamma i_0}  2^{O(t
   \beta)} nU)  {3t/2 \choose t/2}}{ (\varepsilon t)^{(t-2b^2t^{0.9})/4}} = \frac{ n^2
    2^{\gamma i_0}  2^{O(t \beta)}}{ (\varepsilon t)^{t/8}}  \\
    &\le  
    \frac{ n^2
      2^{\gamma \frac{6t\alpha}{\beta}}  2^{O(t \beta)}}{2^{\Omega(t \log\log n)}}   
       \le \frac{ n^2 2^{6 \varepsilon \gamma t}  2^{O(t \beta)}}{2^{\Omega(t \log\log n)}}  \le  \frac{ n^2 }{2^{\Omega(t \log\log n)}}  
  \end{align*}
}
by setting $\varepsilon$ small enough.
This leads to a contradiction since we expected the total area of $\reg(\C)$ over all the main connectors to have area $\Omega(n^2/\log^2 n)$.

Thus, we have proven the following theorem.
\begin{theorem}\label{thm:main}
  If a fully dynamic pointer machine data structure with amortized $O(\log^{O(1)}n)$ update time,
  or an incremental pointer machine data structure with worst-case  $O(\log^{O(1)}n)$ update time
  can answer fractional cascading queries on a subgraph $\scat$ in $O(\log n + \alpha |\scat|)$ time
  in the worst-case, then, we must have $\alpha = \Omega(\sqrt{\LL n})$.
\end{theorem}

%% file: reduction.tex
\section{From Worst-case to Amortized Lower Bounds}\label{app:reduction}
In this section we prove a general reduction that shows under some
conditions, we can generalize a query lower bound for an incremental data structure with 
a worst-case update time to a query lower bound for a fully dynamic data structure but with
the amortized update time.
The exact specification of the reduction is given below. 

%\para{The definition of amortization.}
We work with the following definition of amortization.
We say that an algorithm or data structure has an amortized cost of $f(n)$, for a function $f:\N \to \N$,
if for any sequence of $n$ operations, the total time of performing the sequence of operations
is at most $nf(n)$.

%\para{Epoch-based worst-case incremental adversary.}
%Assume, we are considering an incremental data structure.
We call the following adversary, an Epoch-Based Worst-Case Incremental Adversary (EWIA) with update
restriction $U(n)$; here $U(n)$ is an increasing function.
The adversary works as follows.
We begin with an empty data structure containing no elements and then 
the adversary reveals an integer $k$ and they announce that they will insert $O(n)$ elements
over $k$ epochs.
Next, the adversary allows the data
structure $nU(n)$ time before anything is inserted.
At epoch $i$, they reveal an insertion sequence, $s_i$, of size $n_i$. %, where $n_i < n_{i-1}$.
At the end of epoch $k$, the adversary will ask one query. 
The only 
restriction that the adversary places in front of the data structure 
is that the insertions of epoch $i$ must be done in $n_i U(n)$ time once $s_i$ is revealed;
clearly, any incremental data structure with
the worst-case insertion time of $U(n)$ can do this.
We iterate that the adversary allows the data structure to operate in the stronger pointer machine model 
(i.e., with infinite computational power and full information about the current status of the memory graph).

\reduction*
\begin{proof}
    We use $\A$ to create another data structure $\A'$ that still works in the EWIA model 
    and it satisfied the conditions of our lemma.
    In other words, we would like to create another data structure $\A'$ that can perform
    $k$ epochs of insertions within the allowed time. 
    To get $\A'$, we simulate $\A$ and internally, $\A'$ will constantly perform
    insertions and deletions on $\A$.
    Since $\A'$ will be a data structure that operates in the EWIA model, by the assumptions of the lemma,
    we would know that $\Omega(Q(n))$ is a lower bound for the worst-case query time of $\A'$.
    The query algorithms of $\A$ and $\A'$ will be identical which would imply the lemma.
    We now present the details.

    Observe that the goal of the simulation  is to achieve worst-case insertion time of $U(n) n_i$ at epoch $i$.
    To do that, $\A'$ internally simulates $\A$ and keeps inserting and deleting ``costly'' sequences
    (i.e., sequences with high insertion time)
    until $\A$ reaches a memory configuration where no such sequences exist. 
    However, doing so requires dealing with some non-trivial amount of technical details.

    Let $L$ a parameter that will be determined later. 
    Let $\M$ be the status of the memory of $\A$ at some point during the updates.
    The idea is that, given a memory status, we  can define the ``level 1 cost'' of epoch $i$ as the number of
    ``costly'' insertions and deletions that we need to do until every sequence of 
    $n_i$ insertions has low worst-case insertion cost.
    To formalize, for the memory configuration $\M$, we associate a \mdef{level 1 time} and \mdef{level 1 size}
    using the following mechanism:
    \begin{itemize}
      \item (step 0) Initialize level 1 time and level 1 size of $\M$ to zero.
      \item (loop step) While there exists an insertion sequence $s$ of some size $n_i$, $n_i < n$, 
        such that inserting $s$ takes times $T$, with $T \ge Ln_i$,
        \begin{itemize}
            \item then, insert $s$ and then immediately delete $s$ and 
          increase level 1 time of $\M$ by $T$ and the level 1 size of $\M$ by $n_i$.
        \end{itemize}
    \end{itemize}
    It is not immediately clear that this cost is well-defined since the while loop could actually be an infinite loop.
    However, later, we will pick a value of $L$ that will guarantee that both costs are finite and well-defined.
    Nonetheless, assuming that both costs are finite, observe that if the level 1 size is $X$, then the level 1 time is 
    at least $LX$ and conversely, if the level 1 time is $T$, then the level 1 size
    is at most $T/L$.

    Based on this, and inductively, we define \mdef{level $j$ size} and
    \mdef{level $j$ time} associated with $\M$.
    Intuitively, they define a number of 
    ``costly'' insertions and deletions that we need to do until for every $j$
    sequences $s_i, s_{i+1}, \dots, s_{i+j-1}$
    of sizes $n_i, \dots, n_{i+j-1}$ respectively, inserting each sequence $s_\ell$, $i \le \ell < i+j$ has ``low'' 
    worst-case insertion time.
    This is defined using a slightly more complicated mechanism. 
    \begin{itemize}
      \item (step 0) Assume level $j-1$ time and size of every memory configuration is defined.
      Initialize level $j$ time and level $j$ size of $\M$ to zero.
      \item (loop step) Within a loop, check for the existence of the following two types of expensive insertions;
        if none of them exists, exit the loop.
      \item (first type) If there exists a sequence $s_i$ of some size $n_i$, such that inserting $s_i$
    takes time $T$, with $ T \ge Ln_i$, 
    \begin{itemize}
      \item insert and then delete $s_i$.
          Then, increase level $j$ time of $\M$ by $T$ and the level $j$ size of $\M$ by $n_i$ and go back to the beginning
          of the loop.
    \end{itemize}
  \item (second type) If no sequence of first type exists but there exists a sequence
    $s_i$ of some size $n_i$, such that after inserting $s_i$, the level $j-1$ time, $T$, of
    the resulting memory configuration is at least $n_iLk$, then do the following.
    \begin{itemize}
      \item  Let $X$ and $T$ be the level $j-1$ size and time of the memory configuration after inserting $s_i$,
        respectively.
          We insert $s_i$, then insert and delete all the updates corresponding
          to the level $j-1$ cost of the resulting configuration, and then delete $s_i$.
          We increase the level $j$ size by $n_i + X$ and the level $j$ time by $T$.
          We go back to the beginning of the loop.
    \end{itemize}
\end{itemize}

\begin{lemma}
  If the level $j$ size, $X$, and level $j$ time, $T$, of a memory configuration are finite, 
  then, we have $T \ge XL((k-1)/k))^{j-1}$.
\end{lemma}
\begin{proof}
  We use induction.
  The lemma is trivial for $j=1$ as it directly follows from the definition of level 1 time and size.
  Thus, it suffices to prove the induction step. 
  Consider a memory configuration.
  During each execution of the while loop, the level $j$ time and size are increased by some
  amounts.
  Let $X_1$ be amount of increase in level $j$ size, caused by case (i), which means
  the existence of $X_1$ insertions (and $X_1$ deletions) that take
  at least $LX_1$ time in total.
  Now consider one iteration of case (ii):
  The data structure makes some $n_i$ insertions and let $X'$ and $T'$ be the
  level $j-1$ size and time of the resulting configuration.
  Here, we have $T' \ge n_iLk$ but also by induction hypothesis, $T' \ge X'L(\frac{k-1}{k})^{j-2}$.
  In this case, we have $Y = X' + n_i $ insertions with total cost of at least $T'$.
  Observe that,
  \begin{align*}
    YL(\frac{k-1}{k})^{j-1} &= X'L(\frac{k-1}{k})^{j-1} + n_i L(\frac{k-1}{k})^{j-1} \le T'(\frac{k-1}{k})   + \frac{T'}{k}(\frac{k-1}{k})^{j-1}  \le T'.
  \end{align*}
  As a result, each iteration of while loop creates insertions whose  average insertion cost is 
  at least as claimed, proving the lemma.
\end{proof}

    Set $L = U(n)/(8k)$.
    Now, we try to build an incremental data structure. 
    We begin with an empty data structure. 
    First observe that the level $k$ size  of an empty data structure is bounded;
    a level $k$ of size $X$ implies the existence of $X$ insertions (and $X$ deletions)
    that take at least $ XL((k-1)/k))^{k-1} \ge XL/e$ time. 
    Thus, if $X\ge n$, we have $2X$ updates that take more than $(2X)U(n)/(8k)$ time, a contradiction.
    As a result, we must have $X \le n$. 
    The EWIA model gives $\A'$ full information about the current memory status and unlimited computational power. 
    As a result, $\A'$ can compute the set of insertions and deletions
    defined by the level $k$ size, in advance.
    $\A'$ then simulates $\A$ and performs all the insertions and deletions defined by the level $k$ cost functions.
    After this step, we are guaranteed, by definition, that every sequence $s_1$ of 
    any size $n_1$ has the worst-case insertion time of $Ln_1$ but also crucially,
    after inserting $s_1$, the level $k-1$ cost of the resulting memory configuration
    is bounded by $n_1 kL$. 

    If the adversary chooses to insert $s_1$, $\A'$ first simulates $\A$ and performs the necessary steps to do the
    insertions. Then, it simulates all the insertions and deletions defined
    by the level $k-1$ cost. 
    By the above argument, this takes at most $n_1 kL$ time (or otherwise, inserting $s_1$ would have been
    counted in the level $k$ time of an empty data structure), 
    for a total running time of $n_1 kL + n_1L = n_1 (k+1)L < n_1 U(n)$.
    Doing this, ensures that any sequence $s_2$ of size $n_2$, has 
    the worst-case insertion time of $Ln_2$ but also crucially and similarly, 
    after inserting $s_2$, the level $k-2$ cost of the resulting memory configuration
    is bounded by $n_2 kL$. 

    We can continue this argument inductively, 
    to show that , we obtain an algorithm $\A'$ that guarantees that in the epoch $i$,
    it can use at most $n_i U(n)$ time, regardless of the insertion sequence picked
    by the adversary. 
    As the query algorithm of $\A'$ is identical to $\A$, and since by the adversary
    we have a $Q(n)$ lower bound for the query bound of $\A'$, the same lower bound also
    applies to the query time of $\A$.
\end{proof}